\documentclass[12pt]{article}
\usepackage[T1]{fontenc}
\usepackage{natbib}
\usepackage[english]{babel}
\usepackage{amsmath,amsthm,amssymb}
\usepackage{xfrac}
\usepackage{mathtools}
\usepackage{amsfonts}
\usepackage{chngcntr}
\usepackage{graphicx}
\usepackage{setspace}
\usepackage{enumitem}
\usepackage{fullpage}
\usepackage{gensymb}
\usepackage[dvipsnames]{xcolor}
\usepackage{lscape}
\usepackage{fancyhdr}
\usepackage{floatpag}
\usepackage{microtype}
\usepackage{longtable}
\usepackage{hyperref}
\usepackage{epigraph}
\usepackage[noabbrev,nameinlink]{cleveref}
\usepackage[titletoc]{appendix}
\usepackage{bm}
\floatpagestyle{plain}
\newtheorem{proposition}{Proposition}[section]
\newtheorem{lemma}[proposition]{Lemma}
\newtheorem{corollary}{Corollary}[proposition]
\newtheorem{claim}[proposition]{Claim}

\theoremstyle{definition}

\theoremstyle{remark}

\newtheorem*{remark}{Remark}

\hypersetup{
    colorlinks=true,
    linkcolor=Blue,
    filecolor=Blue,      
    urlcolor=Blue,
    citecolor=Blue,
}
\providecommand{\keywords}[1]{\textbf{\textit{Keywords:}} #1}
\providecommand{\jel}[1]{\textbf{\textit{JEL Classifications:}} #1}

\begin{document}

\title{Competing to Persuade a Rationally Inattentive Agent}

\author{Vasudha Jain\thanks{Corresponding author. Department of Economics, The University of Texas at Austin.  Email: vasudha.jain@utexas.edu.} \and Mark Whitmeyer\thanks{Department of Economics, The University of Texas at Austin.\newline Thanks to V Bhaskar, Vasiliki Skreta, Max Stinchcombe, and Tom Wiseman for helpful feedback and discussions; and to various seminar audiences for helpful comments.} }

\date{\today }%
\maketitle

\thispagestyle{empty}

\begin{abstract}
\noindent Firms strategically disclose product information in order to attract consumers, but recipients often find it costly to process all of it, especially when products have complex features.  We study a model of competitive information disclosure by two senders, in which the receiver may garble each sender's experiment, subject to a cost increasing in the informativeness of the garbling. For a large class of parameters, it is an equilibrium for the senders to provide the receiver's first best level of information--i.e. as much as she would learn if she herself controlled information provision.  Information on one sender substitutes for information on the other, which nullifies the profitability of a unilateral provision of less information. Thus, we provide a novel channel through which competition with attention costs encourages information disclosure.
\end{abstract}
\\
\jel{D82; D83}
\vspace{1in}

\clearpage
\setcounter{page}{1}

\epigraph{Constant attention wears the active mind,\\ 
Blots out our powers, and leaves a blank behind.}{Charles Churchill\\ \textit{Epistle to William Hogarth }}
\section{Introduction}\label{intro}
The standard Bayesian persuasion framework allows senders to design arbitrarily informative signal structures, and assumes that receivers costlessly process all information made available to them. This is an unrealistic assumption in many natural contexts, in which agents may rationally choose to stay partly ignorant. Moreover, there are many situations in which multiple senders compete via information provision to be chosen by the agent. In this competitive scenario, we ask how the consumer's information-processing (or attention) costs shape the information provided by the senders.

Consider, for instance, the situation encountered by doctors and pharmaceutical companies. Patients rely on their doctors to make important medical decisions for them, such as the decision of which medication to take. Often, multiple drugs exist to treat the same condition, but nevertheless differ in subtle ways that can prove crucial for patients. Which alternative is best might depend on the particular circumstances of individual patients--e.g., someone's medical history might make him more prone to the side effects of one of them.

A well intentioned doctor has her task clearly cut out--she should study all primary research published on each drug, and let that information guide her prescription decisions. This means that when she reads about a clinical trial, she should dig into details such as whether, for instance, adverse side effects had led many trial subjects of a certain demographic group to drop out midway, or whether the drug had a differential impact depending on the stage of the illness.

However, getting detailed information involves substantial time and effort, and doctors typically find it difficult to keep up. Tellingly, \cite{alper2004much} find that it would take a doctor six hundred hours to skim all research relevant to general practice that is published in just one month. Consequently, they are likely to pay attention only to some published summary statistics.

Pharmaceutical companies are prohibited from falsifying facts when marketing to doctors. They do, however, strategically decide how much information to reveal and in what form, and in doing so, take into consideration the lack of attention on the part of the recipients: designing pamphlets in a way that the most favorable pieces of evidence stand out, or other strategies of that ilk. As \cite{goldacre2014bad} explains, ``They (doctors) need good quality information, but they need it, crucially, under their noses. The problem of the modern world is not information poverty, but information overload...So doctors will not be going through every trial, about every treatment relevant to their field...They will take shortcuts, and \emph{these shortcuts can be exploited} [emphasis added].''

Motivated by this setting, we study a model of information disclosure \'a la \cite{kamenica2011bayesian}, with two senders, and a receiver who can save on attention costs by drawing from a less informative experiment than what is chosen by the senders. The question we are interested in is how, and to what extent, the degree of attention costs matters for the relationship between competition and information disclosure.

More specifically, our baseline model has two senders who simultaneously commit to a Blackwell experiment for the quality of their respective products, which are \textit{ex-ante} identical. An experiment can simply be identified by a distribution of beliefs that averages to the prior. A receiver, who wishes to choose the sender with a higher quality, visits the senders sequentially. When she visits the first one, she observes the distribution of beliefs induced by his (the sender's) experiment, and is free to choose any mean preserving contraction--or \emph{garbling}--of that. A draw from that garbling determines her posterior belief about that sender. Think back to the doctor example, and the shortcuts she might take: she might read just the first few pages of an article, only the nontechnical parts, only the technical sections, or even just the title. All of these correspond to different levels of information, and all of these impose on the receiver different costs--a grueling slog through a complicated model takes more out of the receiver than does a quick skim of the conversational portions. 

We capture this relationship by imposing that the less informative the receiver's garbling,\footnote{$q$ is less informative than $p$ if $q$ is a garbling of $p$.} the lower her attention costs.\footnote{So, it is not the act of garbling that is costly, but rather the act of drawing more accurate beliefs.} Intuitively, further garbling of an experiment reduces costs because it leads to a belief distribution that is more concentrated around the prior, and hence involves less learning about the state. The receiver has to balance this reduction in attention costs against the worsening quality of information on a decision-relevant variable.

With the first posterior in hand, the receiver decides whether to visit the second sender. Importantly, we do not impose that she must visit the second sender in order to choose him. If she does decide to visit him, the protocol is identical to that for the first sender: she chooses a garbling of his chosen experiment subject to an information cost. Finally, she chooses the sender favored by her posterior beliefs.\footnote{If she does not visit a sender, her posterior is equal to the prior.} Each sender wants to maximize the probability of being chosen.

As we show later on, the receiver's learning strategy has an intuitive feature that drives our analysis: it can happen that she will have ``seen enough'' at the first sender and need not visit the second sender--her belief about the first sender may be so high that she chooses him without ever visiting the second, and it may be so low that she chooses the second, sight unseen. Returning to the example: if a doctor is fairly certain that drug $A$ is of low quality, then she would be willing to prescribe drug $B$ without learning anything about it, and vice versa.

In this setting, a pertinent benchmark is what the receiver would do if she had potential access to full information on each sender, and could effectively use \emph{any} pair of Bayes plausible distributions to learn--subject only to attention costs as described above.\footnote{This is equivalent to an environment where the receiver herself controls how much information is provided by senders.} We show that in this case--the `first best' scenario for her--she would always learn \emph{something} from at least one sender, but never learn any sender's type with \emph{certainty}. Furthermore--and this is crucial--fixing any prior, for a high enough attention cost parameter, it is optimal for her to learn from exactly \emph{one} sender.

The main question we address in this paper is whether in our game with strategic senders, there is an equilibrium in which the receiver ends up with this first best outcome. That is, whether senders voluntarily provide as much information as the receiver would acquire in her first best scenario.

Surprisingly, the answer is yes under general conditions. In particular, for any prior, this holds as long as an attention cost parameter is above a threshold.\footnote{Equivalently, fixing attention costs above a threshold, this holds over an interior interval of prior means. The interval expands as attention costs grow, and approaches the full range as they explode.} This is a departure from closely related models in the literature, where either there is only one sender and the receiver faces attention costs, or there are two senders but no attention costs.

Our analysis produces a sharp economic insight into why a combination of these ingredients--competition and attention costs--gives us more information disclosure. Recall our observation that for high enough costs it is optimal for the receiver to learn from exactly one sender. Now suppose that the sender from whom she does plan to learn, unilaterally deviates and restricts her learning. Then since the other sender continues to provide full information, the receiver could just switch to learning from him instead. Her \textit{ex-ante} payoffs, and the probability of choosing correctly between the senders, would remain unaffected. Since a sender's payoffs ultimately depend only on this probability, the deviation ends up being unprofitable for him.

Our results are robust to a variety of modifications of the model. While we concentrate our attention on experiments that are unobserved by the receiver until her visit (since that fits our leading example best), our main results go through if the experiments are chosen publicly by the senders. As we discover, if anything this strengthens our main results: publicly chosen experiments make it easier for the receiver to see, and hence react to, a sender's deviation. Similarly, for ease of exposition, for most of the paper we focus on the symmetric case in which both senders are identically likely to be the high type. Again, we show that this restriction is not needed for inattention to beget the first-best level of information, and that a result of a similar ilk holds for heterogeneous means (that are not too far apart). Finally, we modify the receiver's information cost function, itself, to allow it to depend on the distribution chosen by the sender. Like the public experiments modification, this only supplements the incentives that drive our main results, which thus continue to hold.

\subsection{Related Literature}

To the best of our knowledge, this paper is the first to look at competitive information design with information processing costs faced by the receiver. This relates thematically to several strands of the literature.

Since in our model the receiver's decision to garble a sender's experiment is the result of an optimization problem that accounts for attention costs, she is \emph{rationally} inattentive as in the economics literature pioneered by \cite{sims2003implications}. The particular framework of rational inattention that we adopt is the same as in \cite{lipnowski2019attention} and \cite{wei2018persuasion}. The former paper considers the problem of a principal whose preferences over actions are perfectly aligned with those of an agent. Attention costs are borne only by the agent, and the authors establish conditions under which the principal would want to restrict her information with a view to manipulating her attention.

\cite{wei2018persuasion} belongs to the small but growing literature on persuasion of a rationally inattentive receiver by a \emph{single} sender. This paper, like ours, considers a binary types, binary action model with a single sender who has state independent preferences, and an exogenous threshold of acceptance for the receiver. It is shown that the sender necessarily finds it in his interest to restrict the receiver's learning--we show how competitive forces change this.

\cite{bloedel2018persuasion} take a different approach to a problem similar to Wei's. In their framework, after observing the sender's experiment, but before seeing its realization, the receiver can choose a mapping from signal realizations to distributions over `perceptions', incurring an entropy reduction cost. Then, the receiver observes the realized perception, and not the actual signal realization. As \cite{lipnowski2019attention} explain, this is conceptually different from our paper (and theirs), since the receiver in our model pays a cost to reduce uncertainty about the \emph{state}, and not the sender's message. \cite{matyskova2018bayesian} studies a persuasion model where the receiver, after observing the realization from the sender's signal, can acquire additional information on the state at a cost proportional to the reduction in entropy.

Our work is also closely related to the papers on competitive information design without any attention costs. With two senders, this has been studied (albeit with slightly different timing than our baseline model) by \cite{cotton}, who identify the unique equilibrium. \cite{Hulko} extend this analysis to $n>2$ senders, while also incorporating the possibility of search frictions. Crucially, providing full information is not an equilibrium with zero attention costs, and we show that this continues to hold for positive but small attention costs. 

Some other papers in the competitive information design literature that bear mentioning are \cite{Au1, Au2}, \cite{Albrecht}, and \cite{boleslavsky2018limited}. The result that competition encourages information disclosure is familiar from some of these,\footnote{The same is also true of some papers in the \emph{cheap talk} literature, e.g. \cite{battaglini2002multiple}, which have a very different flavor.} but as detailed above, introducing attention costs offers a novel perspective on why that might be true. 

\cite{board2018competitive} also look at sellers competing through information to entice buyers. However, there, search is random, the number of sellers is uncountable, and the decision by a seller of how much and what sort of information to disclose is made upon the buyer's visit. Thus, the problem is different from the scenario analyzed here (or in \cite{Hulko}), where the searcher must choose whom to visit as part of a stopping problem. Moreover, in \cite{board2018competitive} the value of each seller's goods to the buyer are perfectly correlated whereas here they are independent. Accordingly, one of the key inputs of their model; how much a seller can observe about the consumer's belief, is absent here.

Another group of papers look at what happens if there is no competition but costs are on the sender's side instead of the receiver's. \cite{gentzkow2014costly} look at optimal persuasion mechanisms when the sender pays higher costs (proportional to entropy reduction) of designing more informative experiments. Likewise, \cite{treusttomala} consider constraints on the sender's information transmission channel. The sender in their paper has $n$ copies of identical persuasion problems, but is constrained to send only $k<n$ messages, which are transmitted with exogenous noise. Interestingly, they find that the sender's payoff from the optimal solution is the concave closure of his payoff function, net of entropy reduction costs. Thus, these costs arise endogenously in their model.

The rest of the paper is organized as follows. Section \ref{model} presents our baseline model. Section \ref{single sender} presents results for the benchmark with a single sender. Section \ref{twosenders} presents the equilibrium analysis with two senders and spells out how the level of attention costs matters. Section \ref{extensions} illustrates the robustness of our results to the various modifications mentioned in the introduction, and Section \ref{conclusion} concludes. The Appendix contains proofs that are not presented in the main text. 

\section{Baseline Model}\label{model}

There are two senders indexed by $i \in \left\{1,2\right\}$, and a receiver ($R$). Sender $i$  has type $\omega_i \in \Omega_i \coloneqq \left\{0,1\right\}$, with the types being drawn independently. The common prior belief is that $\Pr(\omega_i=1)=\mu \in (0,1)$ for $i \in \left\{1,2\right\}$.

$R$ has to select one of the two senders, and she has no outside option.\footnote{Our results hold with an outside option, as long as its expected quality is not too high. } Her payoff is equal to the type of the selected sender, minus \emph{attention costs} that we elaborate on below. Sender $i$'s payoff is $1$ if he is selected, and $0$ if not. All players maximize expected payoffs. 

\subsection{Timing and Strategies}
The game proceeds in the following 3 stages.

\textbf{Stage 0:} Each (\textit{ex-ante} uninformed) sender simultaneously commits to a Blackwell experiment that generates information about his own type. Such an experiment is a mapping from $\left\{0,1\right\}$ to the set of Borel probability measures over a compact metric space of signal realizations. Each signal realization, then, is associated with a posterior belief distribution on $\left\{0,1\right\}$, and an experiment induces a distribution over posterior beliefs. Hereafter, we identify a posterior belief with the belief on $\omega_i=1$.

From the work of \cite{kamenica2011bayesian}, we know that the set of Blackwell experiments is isomorphic to the set of distributions of posterior beliefs whose average is the prior. Thus, at this stage 0, sender $i$ commits to a distribution $p_i \in \Delta[0,1]$, with $\int_{[0,1]} x \  dp_i(x) =\mu$.\\

\textbf{Stage 1:} $R$, who at this point \emph{does not} observe the chosen distributions,\footnote{We relax this assumption in Section \ref{extension1}.} decides whether to visit any sender, and if yes, which one.

Say she visits sender $1$ first. Upon visiting she observes $1$'s distribution $p_{1}$, and is free to choose any $q_{1}$ that is a mean preserving contraction (or garbling) of $p_{1}$.\footnote{$q$ is a garbling of $p$ if the random variable associated with $q$ second order stochastically dominates--and has the same mean as--the random variable associated with $p$. It is a strict garbling if additionally $q \neq p$. Trivially, $q_1=p_1$ or $q_1=\delta_\mu$ is always an option.} Associated with any such $q_{1}$ is an \textbf{attention cost} given by the following:
\begin{equation}\label{costfunction} C(q_{1})= \int_{[0,1]} k(x-\mu)^2 dq_{1}(x),\end{equation}
where $k>0$. Note that costs depend on $q_1$ and not directly on $p_1$. We defer a discussion of these costs to Section \ref{costssection}.

$R$ takes a draw from $q_{1}$, which determines her posterior belief about sender $1$.
\\

\textbf{Stage 2:} $R$ then decides whether to visit sender $2$. If she does, she observes $p_{2}$ and chooses a garbling $q_{2}$, once again incurring an attention cost $C(q_{2})$. She takes a draw from $q_{2}$, which determines her posterior belief about this sender. Finally, she chooses the sender for whom her posterior belief is higher.\footnote{As long as she learns something from at least one sender, the posteriors would never be equal. (See Footnote \ref{footnote}.) They would be equal only if she does not learn from either sender--in that case she may choose between the senders in any way.} She need not have visited a sender or learned anything from him in order to select him.\\

Notice that $R$'s optimal garbling at stage 2 potentially depends on the belief she draws at stage 1. She may be more or less inclined to learn about the second sender, depending on how much uncertainty has already been resolved about the first one. Indeed, as we shall see, if the stage 1 belief is close enough to $0$ or $1$, she chooses not to learn at all at stage 2, and this fact plays a crucial role in our analysis.

The distribution offered by the sender visited first dictates how much can be learned at stage 1. Then in light of the preceding observation, if both senders offer different distributions, the choice of whom to visit first (if anyone) matters for payoffs. \\

A pure strategy for sender $i$ is a choice of a distribution $p_i \in \Delta[0,1]$ whose average is $\mu$. A pure strategy for $R$ consists of i) a choice of which sender to visit first, if any; ii) a choice of garbling for any distribution offered by the sender she visits first; iii) a choice of whether to visit the second sender for each belief drawn in the previous stage; iv) a choice of garbling for the second sender, for each distribution offered by him and each posterior belief drawn in the previous stage.

Our solution concept is subgame perfect equilibrium (hereafter, equilibrium), defined in the standard way. We restrict the players to pure strategies, with one exception--\textit{viz.}, we allow $R$ to randomize over the order of visits.

Before proceeding to our analysis, we point out the following characterization of the set of garblings of a binary distribution, which we shall extensively use: \[q \ \text{is a garbling of a distribution with support} \left\{\nu_1,\nu_2\right\} \iff supp(q) \subseteq [\min\left\{\nu_1,\nu_2\right\},\max\left\{\nu_1,\nu_2\right\}].\]

\subsection{Attention Costs}\label{costssection}
Attention costs, in our framework, are costs incurred to process information on a sender's type. Through his choice of a Blackwell experiment, a sender can control how much information on his type is available--in other words, he can put a cap on what can be learned. But a recipient may choose to ignore some of that information and take a draw from a less informative experiment, thereby reducing attention costs. For instance, a pharmaceutical company can decide how much research on its drug to publish, but a doctor might choose to read a subset of that. Her costs would depend on how much of the research she chooses to read, not on how much was published. In particular, the act of garbling \textit{per se} is costless.\footnote{As discussed in the section on related literature, this is closely related to the framework in \cite{lipnowski2019attention} and \cite{wei2018persuasion}.}

The cost function in equation \ref{costfunction} captures this notion. Associated with each posterior $x$ is a cost $k(x-\mu)^2$, so that  more accurate beliefs--those that are further away from the prior--cost more. This is integrated to determine the cost of a distribution of posteriors. Thus, costs are \emph{posterior separable} as in \cite{caplin2019rationally}.

Since $k(x-\mu)^2$ is strictly convex, by Jensen's inequality we have
\[  q \ \text{is a garbling of} \ p \implies C(q) \leq C(p),\] with the inequality strict for strict garblings.\footnote{Of course, this property holds for any strictly convex function instead of $k(x-\mu)^2$, but we work with the specific form for tractability.}
For instance, $C(q)$  is minimized for the uninformative distribution $\delta_\mu$, and maximized for the fully informative one with support $\left\{0,1\right\}$.

Clearly then, $R$ faces a trade-off in her choice of garblings $q_{1}$ and $q_2$--a garbling costs less, but also corresponds to a less informative experiment and is less valuable for her decision problem (\citealt{blackwell, blackwell1953equivalent}). Returning to our example, the more extensive or detailed the research a doctor chooses to read, the costlier it is to draw an inference from it; but also, the more confidence she can place in that inference.

\section{Benchmark: Single Sender}\label{single sender}
We begin by taking a brief look at what happens if there is only one sender. $R$ chooses a garbling of that sender's distribution and accepts his product if the belief drawn from it is above a threshold $\lambda \in (0,1)$. Payoffs clearly depend only on the distribution finally chosen by $R$. In any (subgame perfect) equilibrium, the sender offers a distribution to maximize his expected payoff, correctly anticipating $R$'s optimal garbling behavior. 

Following the arguments in \cite{wei2018persuasion}, the setup with a single sender permits two simplifications that are not valid in our two-sender model. One: to obtain the set of equilibrium outcomes, it is without loss of generality to restrict the sender to incentive compatible distributions--those that $R$ would not want to garble. This leads to the second simplification, which is that it is without loss to restrict him to binary and degenerate distributions. (Since $R$ has only two actions, she never wants to pay to generate more than two beliefs.)

We refer to a distribution offered by the sender in a sender-preferred equilibrium as \emph{sender-optimal}.

If $\lambda<\mu$, it is immediate that any sender-optimal distribution is such that nothing is learned and he is accepted with certainty. For example, he can simply design an uninformative experiment.

Now say $\lambda>\mu$. If $k=0$, then we have a standard Bayesian persuasion problem, and we know from prior work that the sender-optimal distribution has support $\left\{0,
\lambda\right\}$. But if $k>0$, this is no longer optimal, because the garbling chosen by $R$ in response to that would be $\delta_\mu$, and the sender would not be accepted. This is easy to see intuitively--at a belief $\lambda$, $R$ is indifferent between accepting and rejecting. When offered $\left\{0,\lambda\right\}$, her gross payoff from choosing any garbling is the same as the payoff from rejecting with certainty. But then there is no reason for her to pay to learn anything. To make it worth her while to do so, the sender would have to allow her to generate beliefs above $\lambda$.

The following proposition summarizes the results for this benchmark case.

\begin{proposition}\label{prop1}

Suppose there is a single sender, and $R$ has a threshold of acceptance $\lambda>\mu$. Then,
\begin{enumerate}
    \item A sender-preferred equilibrium exists.
    \item In a sender-preferred equilibrium, $R$'s garbling on path is strictly less informative (in the Blackwell sense) than her optimal garbling in response to full information.
\end{enumerate}

\end{proposition}

\begin{proof}
See Lemma 1 and Proposition 2 in \cite{wei2018persuasion}.
\end{proof}

Importantly, as will be shown ahead in this paper, $R$ has a unique optimal garbling in response to any binary distribution, which implies that this result holds for \emph{any} equilibrium and we can omit the `sender-preferred' qualifier.

\begin{corollary}
If there is a single sender, and $R$ has a threshold of acceptance $\lambda \in (0,1)$, then full information is not offered by the sender in equilibrium.
\end{corollary}

In response to full information, say $R$'s optimal garbling has support $\left\{\nu_1,\nu_2\right\}$ where $\nu_1<\lambda<\nu_2$. Then these results state that in equilibrium, the distribution $R$ ends up with would be a strict garbling of this. Stated differently, in equilibrium the sender does restrict $R$'s learning, \emph{not allowing her to choose her first best}.

The intuition roughly is that although the sender cannot implement his optimal no-garbling solution $\left\{0,\lambda\right\}$, he need not go all the way to providing full information. He can profitably restrict learning so that the higher belief in the support of $R$'s garbling is below $\nu_2$, and the probability of its realization is higher.

As we see next, introducing an additional sender yields an interesting comparison to this.

\section{Equilibrium Analysis with Two Senders}\label{twosenders}
We now analyze the game described in Section \ref{model}, for an arbitrary $k>0$ and $\mu \in (0,1)$.

To start off, recall our observation that $R$'s order of visits matters when the two distributions on offer are different. In equilibrium she must correctly anticipate the distributions chosen, and the order of visits must be a best response to those. However, since she does not observe the chosen distributions at stage 0, any deviation by a sender goes undetected until and unless he is visited. This has the following implication, which we note for further reference.

\begin{remark}
Any deviation by a sender cannot affect either $R$'s decision to visit a sender, or the order of her visits.
\end{remark}

Next, note that if both senders offer the same distribution, then $R$ is indifferent between the two orders of visit (if she visits anyone). The analysis below will make it clear that the tie breaking rule in this case does not matter for our results, and we do not assume anything about it.

We now turn to the question of equilibrium existence. Suppose that each of the two senders offers no information, i.e. the distribution $\delta_\mu$. Then upon visiting either sender, $R$ is also restricted to choosing $\delta_\mu$. But then she expects to gain nothing by visiting a sender, and not visiting either of them is a best response. She may simply select sender 1 with any probability $p \in [0,1]$, and sender 2 with probability $1-p$. Clearly, if this best response is played, a deviation by a sender goes undetected, and does not make any difference to the outcome. Thus we have the following.

\begin{claim}[Equilibrium existence]
An equilibrium exists for all $\mu \in (0,1)$, and for all $k>0$. In particular, there is always an equilibrium in which each sender offers an uninformative distribution.
\end{claim}

Naturally, we are interested in finding other, more interesting equilibria. Of particular interest are equilibria that give $R$ her first best payoff.

First, let us clarify what exactly we mean by this. $R$'s first best payoff is essentially the best she can achieve, across all profiles of sender behavior. In other words, it is the payoff she would get if she herself could choose the senders' distributions. Now, since every Bayes-plausible distribution is a garbling of the fully informative distribution, she has greatest latitude when both senders offer the fully informative distribution. Thus, her first best payoff is attained when both senders offer full information.

However, she may attain the same payoff even when senders choose other less informative distributions. Suppose, for illustration, that when offered full information, the following is a best response for $R$: visit Sender 1 first, choose the garbling $\left\{\mu-\epsilon,\mu+\epsilon\right\}$ for him; then visit Sender 2 and choose the uninformative garbling for him. Then even if, e.g., Sender 1 offers the distribution $\left\{\mu-2\epsilon,\mu+2\epsilon\right\}$ and Sender 2 offers no information, she gets to choose the exact same response and secure her first best payoff. At the heart of this is the fact that due to attention costs, she might not (and as we shall see, does not) really use full information even when allowed to.

The next observation is easy to make.

\begin{remark}
Suppose there is an equilibrium in which Sender $i$ offers $p_i$. Then $R$ achieves her first best payoff in this equilibrium if and only if her best response on path is also a best response on path to full information from both senders.
\end{remark}

This observation is used to prove the `only if' direction of the following proposition--the argument draws on some results from later sections in the paper, and is presented in the Appendix. The `if' direction is obvious.

\begin{claim}[First best]\label{firstbestclaim}
For given parameters, there is an equilibrium that gives $R$ her first best payoff if and only if there is an equilibrium in which both senders offer full information. 
\end{claim}

An implication of this is that in establishing conditions for the existence of a full information equilibrium, we establish conditions for an equilibrium where $R$ gets her first best. Hence, we focus our attention on full information, and the next proposition presents our main result.

\begin{proposition}[Full information equilibrium]\label{main}
\begin{enumerate}
    \item For all $k>\frac{1}{2}$, there \emph{is} an equilibrium in which both senders offer full information if and only if $\mu \in \left[\frac{1}{4k},1-\frac{1}{4k}\right]$.
    \item For all $k \in \left(0,\frac{1}{2}\right], \mu \in (0,1)$, there is no equilibrium in which both senders offer full information.
\end{enumerate}

     \end{proposition}

It is worth highlighting that the parameter in the attention cost function is crucial. If $k$ is above $\frac{1}{2}$, we obtain an interval of priors over which full information is an equilibrium, and this interval expands as $k$ grows. In the limit, as $k \rightarrow \infty$, the interval converges to $(0,1)$, the full range of priors. Thus, by having higher attention costs, the receiver might elicit better information from competing senders.

The following corollary states the same result differently.

\begin{corollary}
For all $\mu \in (0,1)$, there is an equilibrium in which both senders offer full information if and only if $k > \max\left\{\frac{1}{4\mu},\frac{1}{4(1-\mu)}\right\}$ (weak inequality if $\mu \neq \frac{1}{2})$.
\end{corollary}

Stated this way, one might conjecture that the result is trivially obtained because for high enough values of $k$, $R$ finds it optimal to not learn anything at all even when offered full information. As it turns out, this is not the case, and for any finite $k$ she does undertake some learning from at least one sender when offered full information.

Instead, we obtain the existence result because for high enough values of $k$, $R$ finds it optimal to learn only about the quality of one sender, and completely ignore information on the other. The analysis ahead will clarify how this fact plays a crucial role.

Section \ref{kpositive} provides a proof of this result (and presents additional results), but before we move on to that, it is instructive to examine another benchmark, where $k=0$.

\subsection{Benchmark: No Attention Costs \boldmath (\texorpdfstring{$k=0$}{tex}) \unboldmath}
When $k=0$, it is costless for $R$ to learn. This makes a substantive difference to the analysis, because she never has a strict incentive to garble either sender's distribution--there is no reason to leave any information on the table.\footnote{This setup has been studied in \cite{boleslavsky2018limited}, \cite{Hulko} and other papers, with the difference that they assume that $R$ observes the chosen distributions at stage 0.} Notice that in this case $R$'s first best payoff is obtained \emph{only} when both senders offer full information.

For simplicity, here we assume the following tie-breaking rules i) if the stage 1 draw is $0$ (or $1$), she rejects (or accepts) that sender without visiting the other one, and ii) if the draws from both stages are the same, she selects the sender visited last.

\begin{proposition}[No attention costs]\label{kzero}
Suppose $k=0$. Then the following are true.
\begin{enumerate}
    \item For all $\mu \in (0,1)$, there is an equilibrium in which both senders choose an uninformative distribution.
    \item For all $\mu \in (0,1)$, there is no equilibrium in which both senders offer full information.
\end{enumerate}
\end{proposition}

The reason an uninformative equilibrium exists is identical to that for $k>0$--it is a best response for $R$ to not visit either sender, but then a deviation is not detected and makes no difference to the outcome. The reasoning behind non-existence of a full information equilibrium, on the other hand, is very different for $k=0$ and for small, positive $k$.

For $k=0$, in response to full information from both senders, $R$ visits either one of them, learns his type perfectly, and immediately takes a decision. A sender's deviation cannot make a difference if he is not visited first. But if he is, a deviation to support $\left\{\epsilon,1\right\}$ is profitable, where $\epsilon$ is arbitrarily close to zero. This is because if $R$'s draw from this distribution is $\epsilon$, she continues to learn from the second sender, and rejects him if the draw then is $0$.

When $k$ is any positive quantity, a deviation of this nature does not help--intuitively, even if the stage 1 draw is a small, positive $\epsilon$, $R$ is sure enough of the quality of the first sender that she does not find it worth her while to learn about the other one. 

The following result establishes existence of other (less than fully) informative equilibria when attention costs are absent.

\begin{claim}\label{kzeroother}
\begin{enumerate}
    \item Let $k=0$ and $\mu \leq \frac{1}{2}$. There is an equilibrium in which each player chooses the uniform distribution on $[0,2\mu]$.
    \item Let $k=0$ and $\mu > \frac{1}{2}$. There is an equilibrium in which each sender chooses a CDF with a continuous portion $F(x) = \frac{x}{2\mu}$ on $[0,2(1-\mu)]$ and a point mass of size $2 - \frac{1}{\mu}$ on $1$. In such an equilibrium, $R$'s decision about whom to visit first must be fair (each sender is visited first with probability $\frac{1}{2}$).
\end{enumerate}
\end{claim}

\subsection{Positive Attention Costs \boldmath (\texorpdfstring{$k>0$}{tex}) \unboldmath}\label{kpositive}
As previously discussed, for positive attention costs our main result pertains to the full information equilibrium, which is stated in Proposition \ref{main} above. We begin by showing why it is true, and for ease of exposition present the key arguments for $k=1$. The structure of the proof is the same for a generic $k>0$, and the details are relegated to the Appendix. In essence, the argument will be that full information is an equilibrium when $R$ wishes to learn only from one sender.

\subsubsection{Full information equilibrium for \boldmath \texorpdfstring{$k=1$}{tex} \unboldmath}
Recall that for $k=1$, Proposition \ref{main} states that full information is an equilibrium if and only if $\mu \in \left[\frac{1}{4}, \frac{3}{4}\right]$.

Start by considering any $\mu \in (0,1)$, and suppose that each sender offers the fully informative distribution with support $\left\{0,1\right\}$. To analyze $R$'s (on path) best response, we proceed in two steps--first, we determine her stage 2 best response for each belief drawn at stage 1; second, we use that to solve for the optimal stage 1 behavior. We make use of the technique of concavification for this.\\

\boldmath $R'$\unboldmath\textbf{s stage 2 best response:} First let us find the optimal stage 2 garbling, if $R$ visits the sender at that stage.

Say the draw from Stage 1 is $x \in [0,1]$. Then, $R$ selects the second sender if and only if the stage 2 draw $y$ turns out to be higher than $x$.\footnote{It does not matter what we assume about the tie breaking rule when $y=x$. For any distribution offered, $x$ would not belong to the support of the garbling chosen at stage 2. The reasoning is similar to the argument for why the standard Bayesian persuasion solution is not incentive compatible in the single-sender case (see Section \ref{single sender}).\label{footnote}} Her payoff from a stage 2 belief $y$ is then $max\left\{x,y\right\}$, minus the attention cost associated with $y$. Denote this stage 2 payoff by $U_2(y;x)$.

\[U_2(y;x)=max\left\{x,y\right\}-(y-\mu)^2,\]
for $x,y \in [0,1]$. This is piecewise concave in $y$, and Figure \ref{stage2} plots it for a representative value of $x$.

\begin{figure}[ht]
\centering
\includegraphics[scale=0.2]{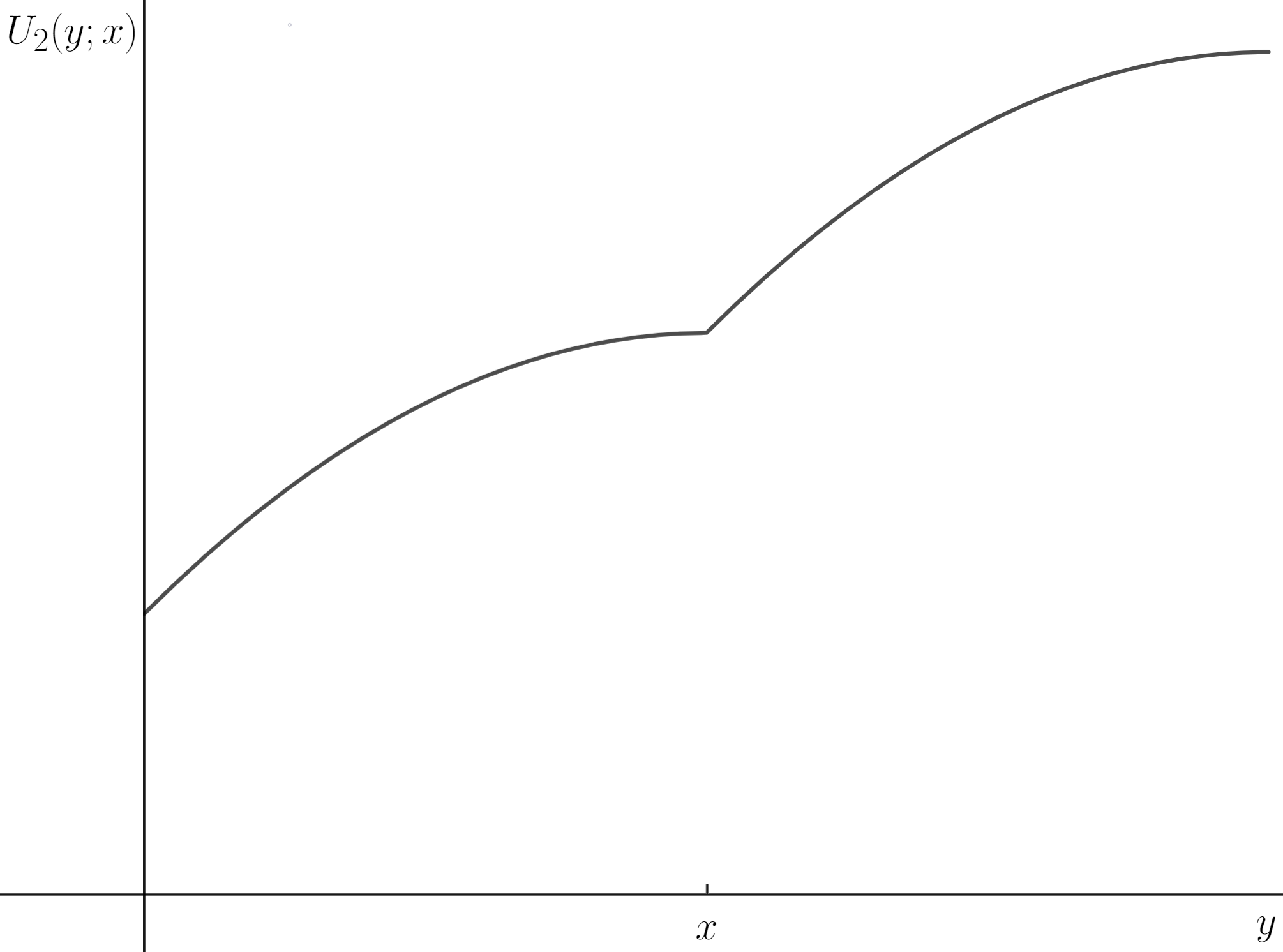}
\caption{$R$'s stage 2 payoffs}
\label{stage2}
\end{figure}

Now, since any distribution is a garbling of the one with support $\left\{0,1\right\}$, we know from \cite{kamenica2011bayesian} that for any given $x$, $R$'s optimal garbling is determined using the concavification of $U_2(y;x)$ over $[0,1]$. The concavification is the red curve in Figure \ref{concav}. It is evident that depending on where $\mu$ lies, the optimal distribution of beliefs is either degenerate on $\mu$, or is binary.

\begin{figure}[ht]
\centering
\includegraphics[scale=0.2]{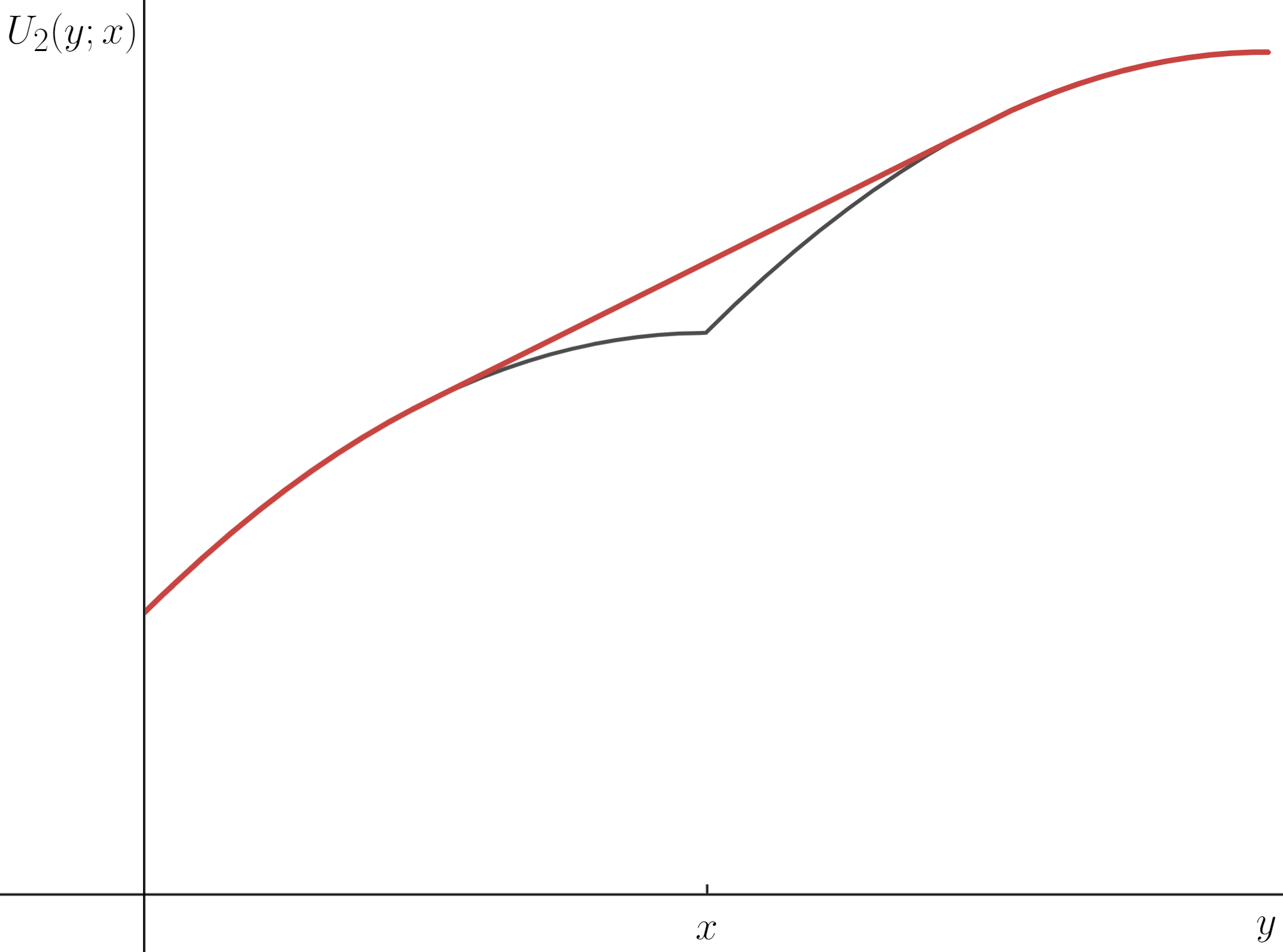}
\caption{Concavification of $R$'s stage 2 payoffs}
\label{concav}
\end{figure}

\begin{lemma}[Stage 2 optimal garbling]\label{distristage2}
Suppose that $R$'s stage 1 draw is $x \in [0,1]$ and she visits the sender at stage 2. $R$'s stage 2 optimal garbling is either degenerate or binary, and its support is as follows.
\begin{enumerate}
    \item If $\mu<\frac{1}{2}$, \[\begin{cases} \left\{x - \frac{1}{4}, x + \frac{1}{4}\right\} & if \  \frac{1}{4} \leq x < \mu + \frac{1}{4}\\ \left\{0,\sqrt{x}\right\} & if \ \mu^2<x< \frac{1}{4}\\ \left\{\mu\right\} & if x \leq \mu^2 \ or \ x \geq \mu + \frac{1}{4}  \end{cases}\]

\item If $\mu=\frac{1}{2}$, \[ \begin{cases} \left\{x - \frac{1}{4}, x + \frac{1}{4}\right\} & if \ \frac{1}{4} < x < \frac{3}{4}\\ \left\{\mu\right\} & if \ x \leq \frac{1}{4} \ or \ x \geq \frac{3}{4} \end{cases}\]

\item If $\mu>\frac{1}{2}$, \[\begin{cases} \left\{x - \frac{1}{4}, x + \frac{1}{4}\right\} & if \ \mu-\frac{1}{4}<x \leq \frac{3}{4}\\ \left\{1-\sqrt{1-x},1\right\} & if \ \frac{3}{4}<x<1-(1-\mu)^2\\ \left\{\mu\right\} & x \leq \mu-\frac{1}{4} \ or \ x \geq 1-(1-\mu)^2 \end{cases}\]
\end{enumerate}
\end{lemma}

The interesting thing to note here is that regardless of the prior, if the first stage draw is either very high or very low, then $R$ chooses not to learn anything from the second sender. This is intuitive--for a high enough belief that the first sender's quality is good, she deems it very unlikely that the second sender is better, and does not invest in learning about him. Instead, she accepts the first sender with certainty. Conversely, if the first stage draw is very low, she accepts the second sender with certainty.

Furthermore, the thresholds beyond which there is no learning at stage 2 depend on the prior. The prior is the expected quality of the second sender, so the higher it is, the larger (smaller) the range of first stage beliefs over which the second sender is accepted (rejected) without learning.

If $R$ does choose a binary distribution at stage 2, then she selects the second (first) sender at the higher (lower) belief.

For any stage 1 draw, if the stage 2 optimal garbling involves any learning, $R$ strictly gains from visiting the second sender. If it does not involve any learning, $R$ is indifferent between making the second visit and not, and she may resolve this in any manner.\\

\textbf{\boldmath $R$\unboldmath's stage 1 best response:} Using the above result, it is straightforward to obtain $R$'s first stage continuation payoffs for an arbitrary $x$, and determine her first stage optimal garbling from its concavification over $[0,1]$. This leads to the following.

\begin{lemma}[Stage 1 optimal garbling]\label{distristage1}
Any Bayes plausible distribution with support drawn from the following sets is optimal for $R$ at stage 1.
\begin{enumerate}

\item $\left\{\mu - \frac{1}{4}\right\} \cup \left[\frac{1}{4}, \mu + \frac{1}{4}\right]$ if $\mu \in \left[\frac{1}{4},\frac{1}{2}\right]$.

\item $\left[\mu - \frac{1}{4}, \frac{3}{4}\right] \cup \left\{\mu+\frac{1}{4}\right\}$ if $\mu \in \left[\frac{1}{2},\frac{3}{4}\right]$.

\item $\left\{0,y_1(\mu)\right\}$ if $\mu< \frac{1}{4}$, where $y_1(\mu) \in \left(\mu,\frac{1}{4}\right)$.

\item $\left\{y_2(\mu),1\right\}$ if $\mu>\frac{3}{4}$, where $y_2(\mu) \in \left(\frac{3}{4},\mu\right)$.

\end{enumerate}
\end{lemma}

The exact expressions for $y_1(\mu)$ and $y_2(\mu)$ are not important. The main thing to note here is that the stage 1 solution always involves some learning, and is unique if and only if $\mu \not \in \left[\frac{1}{4}, \frac{3}{4}\right]$. Interestingly, in spite of the fact that there are only two senders and binary types in this model, $R$ may choose to generate more than two beliefs at stage 1. The reason is that each stage 1 belief is optimally followed by a different degree of learning at stage 2.

Note also that since the stage 1 optimal distribution always involves learning, a visit is necessarily made at this stage. $R$ does not care which sender is visited first, and she may randomize her choice in any way. 

Since there are multiple best responses for $\mu \in \left[\frac{1}{4}, \frac{3}{4}\right]$, we need to make a selection among them. Notice that the \emph{most informative} (in the Blackwell sense) of the optimal distributions has support $\left\{\mu-\frac{1}{4},\mu+\frac{1}{4}\right\}$, and by Lemma \ref{distristage2}, this is the only one among them that is necessarily followed by no learning at stage 2. We assume that $R$ breaks her indifference in favor of this distribution. That is, when indifferent, she'd rather not put off learning until the next stage.\\

In summary: if $\mu \in \left[\frac{1}{4}, \frac{3}{4}\right]$, $R$'s best response to full information is the following. Visit sender 1 with probability $q \in [0,1]$, and Sender 2 with probability $1-q$. Choose the garbling with support $\left\{\mu-\frac{1}{4},\mu+\frac{1}{4}\right\}$ for the sender visited. If the belief drawn is $\mu-\frac{1}{4}$, select the other sender without learning anything from him.\footnote{Either by not visiting him at all, or by visiting but not learning.} If the belief drawn is $\mu+\frac{1}{4}$, select the visited sender without learning anything from the other one.\footnote{ibid.}\\

We now know what happens on path if full information is offered. For $\mu \in \left[\frac{1}{4}, \frac{3}{4}\right]$, it turns out that we can rule out profitable deviations without exactly knowing $R$'s best response to any deviation. For $\mu \not \in \left[\frac{1}{4}, \frac{3}{4}\right]$, we show that there exists a profitable deviation for a sender.\\

\textbf{No profitable deviation for \boldmath $\mu \in \left[\frac{1}{4}, \frac{3}{4}\right]$ \unboldmath:} Consider what a sender achieves by deviating. We have already seen that this does not affect the probability of being the one to be visited first. Moreover, if he is \emph{not} the one to be visited first, his payoffs are not affected, since $R$ does not plan to learn anything from him. So, we only need to consider what happens if he deviates and is visited first. In this case, $R$'s behavior would be altered if $\left\{\mu-\frac{1}{4},\mu+\frac{1}{4}\right\}$ is \emph{not} a garbling of the distribution he deviates to.

Now, say $R$ visits a sender and finds out that she may no longer choose support $\left\{\mu-\frac{1}{4},\mu+\frac{1}{4}\right\}$. Regardless of what the sender's deviation is, though, she is permitted to learn nothing, i.e. choose support $\left\{\mu\right\}$. By Lemma \ref{distristage1}, this is one of her best responses, and by Lemma \ref{distristage2}, this would be optimally followed by visiting the other sender (who has not deviated) and choosing support $\left\{\mu-\frac{1}{4},\mu+\frac{1}{4}\right\}$ for him.

By responding to the deviation in this manner, $R$ ensures a payoff equal to what is attained in the absence of the deviation. Naturally, any other response specified in Lemma \ref{distristage1}, if permissible under the deviation, would also give her the same payoff, and she may choose that instead of support $\left\{\mu\right\}$.

What this essentially implies is that in response to any deviation by the sender visited first, $R$ would choose from the set specified in Lemma \ref{distristage1}, and depending on the belief she draws, follow it with stage 2 behavior specified in Lemma \ref{distristage2}.

This observation, and the next Lemma, are key to our analysis.

\begin{lemma}\label{fullinfoeqm}
For all $\mu \in \left[\frac{1}{4}, \frac{3}{4}\right]$, conditional on being visited first, a sender's expected payoff is the same for any of the receiver responses specified in Lemma \ref{distristage1}.
\end{lemma}

This immediately implies that a unilateral deviation does not affect a sender's payoffs, and it is proven that full information is an equilibrium for $\mu \in \left[\frac{1}{4}, \frac{3}{4}\right]$.

Let's take a closer look at the intuition behind this. We see that in the best case scenario for $R$, i.e. when both senders allow her perfect information, attention costs lead her to learn from only one sender. It does not matter which sender that is--if it is a best response to learn only from the sender visited first, then clearly it is also a best response to learn only from the one visited second.

Now, we say that on path she chooses to learn from the first sender she visits. If, however, that sender deviates and restricts her learning, she is able to compensate for it by learning more from the other sender. The \textit{ex-ante} probability that she makes the correct choice thereby remains unaffected, and the deviating sender is unable to gain. \\

\textbf{Existence of profitable deviation for \boldmath$\mu \not \in \left[\frac{1}{4}, \frac{3}{4}\right]$\unboldmath:}
When $\mu \not \in \left[\frac{1}{4}, \frac{3}{4}\right]$, similar reasoning does not apply, since $R$'s best response is unique and involves learning from both senders on path. In this case, there exists a deviation where a sender profitably restricts $R$'s learning in case he is visited second, without affecting what happens if he is visited first.

In particular, say for instance $\mu< \frac{1}{4}$. Recall that in response to full information, $R$ chooses support $\left\{0,y_1(\mu)\right\}$ at stage 1. Following belief $0$ she immediately accepts the second sender, and following belief $y_1(\mu)$, she chooses support $\left\{0,\sqrt{y_1(\mu)}\right\}$ at stage 2.

It can be shown that $\exists p_2 \in (y_1(\mu),\sqrt{y_1(\mu)})$ such that if a sender deviates to $\left\{0,p_2\right\}$ and is visited second (by $R$ holding a belief $y_1(\mu)$ from stage 1), $R$ chooses $\left\{0,p_2\right\}$ instead of $\left\{0,\sqrt{y_1(\mu)}\right\}$. If he is instead visited first, $R$'s best response is unchanged. Evidently this deviation increases the probability of being selected, and is therefore profitable.

\subsubsection{Other equilibria for \boldmath \texorpdfstring{$k>0$}{tex} \unboldmath}\label{otherbinary}
The analysis so far tells us that for any $k$, first, an uninformative equilibrium always exists; and second, a full information equilibrium exists for a large class of parameter values.

Our focus on full information is not misplaced, in spite of the fact that $R$ never makes use of it even when it is on offer.\footnote{We showed this for $k=1$, but Appendix \ref{prooflong} shows that it is true generally: in response to full information, she visits only one sender and picks support $\left\{\mu-\frac{1}{4k},\mu+\frac{1}{4k}\right\}$.} The reason we focus on it--as we have discussed at length--is Claim \ref{firstbestclaim}: full information is an equilibrium if and only if $R$ can get her first best outcome in an equilibrium.

Our interest in equilibria where $R$ gets her first best is natural--first, the existence of such equilibria is surprising; second, in many situations it is appropriate to select receiver-preferred equilibria.

This brings us to two questions. First, what can we say about other equilibria (besides full information) that give $R$ her first best payoff? And second, what can we say about equilibria that do not give her the first best?

The answer to the first question is easy. There is a whole class of such equilibria, where the distributions offered by the senders allow her to respond exactly as under full information.

\begin{claim}\label{obvious}
Suppose $k>\frac{1}{2}$ and $\mu \in \left[\frac{1}{4k},1-\frac{1}{4k}\right]$. For $i \in \left\{0,1\right\}$, let $p_i \in \Delta[0,1]$ be any distribution with expectation $\mu$, and of which the distribution with support $\left\{\mu-\frac{1}{4k},\mu+\frac{1}{4k}\right\}$ is a garbling. Then, there is an equilibrium in which sender $i$ offers the distribution $p_i$. Such an equilibrium is outcome equivalent to full information.
\end{claim}

This can be used to construct specific examples such as the following.

\begin{corollary}\label{uniformkpositive}
\begin{enumerate}
    \item Let $\mu \leq \frac{1}{2}$. Then there is an equilibrium in which both senders offer the uniform distribution on $[0,2\mu]$ if $k \geq \frac{1}{2\mu}$.
    \item Let $\mu > \frac{1}{2}$. Then there is an equilibrium in which both senders offer a CDF with a continuous portion $F(x) = \frac{x}{2\mu}$ on $[0,2(1-\mu)]$ and a point mass of size $2 - \frac{1}{\mu}$ on $1$ if $k \geq \frac{1}{2\mu}$ for $\mu \leq \frac{2}{3}$, and if $k \geq \frac{1}{4(1-\mu)}$ for $\mu \geq \frac{2}{3}$.
\end{enumerate}
\end{corollary}

We choose to present these examples for a reason: Recall from Claim \ref{kzeroother} that these distributions are also equilibria in the $k=0$ scenario, where full information is \emph{not} an equilibrium. In contrast, here these equilibria are in fact outcome equivalent to full information. The difference arises since with attention costs, both full information and these distributions are garbled down to the same thing by $R$.

Turning to the question on equilibria that do not give $R$
her first best, we have the following sharp result that rules out their existence in the class of \emph{binary, symmetric} equilibria.\footnote{That is, equilibria where both senders offer the same distribution, that has binary support.} Essentially, it implies that if a symmetric binary equilibrium exists (for any parameters), so must the full information equilibrium, and in fact it must be outcome equivalent to the full information one--so that $R$ must be getting her first best.
 
\begin{proposition}\label{binarygeneral}
Let the distribution $p$ have support $\left\{l,h\right\}$ with $l \in [0,\mu)$ and $h \in (\mu,1]$.
\begin{enumerate}
    \item If $k>\frac{1}{2(h-l)}$, there is an equilibrium where both senders offer $p$ if and only if $\mu \in \left[l+\frac{1}{4k},h-\frac{1}{4k}\right]$.
    \item If $k \leq \frac{1}{2(h-l)}$, there is no equilibrium where both senders offer $p$.
\end{enumerate}
\end{proposition}

The proof uses arguments similar to those for the full information equilibrium.\footnote{For $h=1$, $l=0$ this proposition is identical to Proposition \ref{main}.} 

Beyond this, characterizing \emph{all} equilibria of the game is beyond the scope of this paper, a major reason being that little practical can be said about the set of garblings of a non-binary distribution.

\section{Extensions}\label{extensions}

Here, we illustrate that our main results continue to hold under a variety of different modelling choices. We begin by demonstrating that they continue to hold if instead the senders' experiment choices are public.

\subsection{Publicly Chosen Experiments}\label{extension1}

So far we have considered a scenario where $R$ does not observe the experiment chosen by a sender until she visits him. This is a natural assumption in many situations--e.g., pharmaceutical companies have much leeway in dissemination of information to doctors, and it may not be known how detailed the research published on a drug is until one takes a look through it. As another example, one does not know how many customer reviews a seller has allowed to be posted on his website until one visits the website.

However, there can be other applications where one might expect the senders' experiments to be observed when they are chosen, i.e. at Stage $0$ of our game. In this case, the strategic considerations remain the same except for an important difference--posted experiments allow a sender's deviation to be observed by $R$, and hence deviations affect her order of visits. For example, if both senders are expected to choose the same experiment, then $R$ is indifferent between the order of visits; by deviating, a sender can break this indifference.

We show that our main results continue to hold in this scenario.

\begin{proposition}\label{publicexp}
Suppose that the experiments chosen by the senders are observed by $R$ at Stage 0 of the game. Then Proposition \ref{binarygeneral} still holds.
\end{proposition}

Note that Proposition \ref{binarygeneral} encapsulates our main results: it establishes conditions for existence of a full information equilibrium, and proves non-existence of any binary, symmetric equilibrium that is not outcome equivalent to full information.

A full proof of Proposition \ref{publicexp} follows in the Appendix, but by recalling the arguments that lead up to existence of a full information equilibrium for some parameters in the baseline model, it is easy to see why that result is robust: If $R$'s best response to full information in the baseline model is to learn from exactly one sender, then a deviation cannot affect payoffs even if it is observed at Stage $0$. Put succinctly, if $R$ knows that a sender has deviated, she can simply visit the \emph{other} sender and learn from him.

\subsection{Different Means}

Our baseline model assumes that the distributions of the senders' types have identical means. There, the receiver's problem is the most interesting, since \emph{ex-ante} she has very little information to base her choice on. 

In this section, we relax this assumption. Our aim is to show that our main result applies to settings where the prior beliefs on the two senders are different but sufficiently close. Namely, if the two means $\mu_{1}$ and $\mu_{2}$ lie in a particular interval, then there is an equilibrium in which both senders offer full information. 

For expositional convenience let $k = 1$. We proceed in the same vein as in the proof for Lemma \ref{distristage1}: using Lemma \ref{distristage2} we can obtain the receiver's first stage continuation payoffs for an arbitrary first stage draw $x$, which we then concavify to obtain the receiver's first stage optimal garbling. Suppose that sender $2$ is visited second (we'll revisit this assumption shortly).

\begin{lemma}[Stage 1 optimal garbling]\label{diffmean1}
Any Bayes plausible distribution with support drawn from the following sets is optimal for $R$ at stage 1.
\begin{enumerate}

\item $\left[\mu_{2}-\frac{1}{4}, \frac{3}{4}\right] \cup \left\{\mu_{2}+\frac{1}{4}\right\}$ if $\mu_{2} \in \left[\frac{1}{2},\frac{3}{4}\right]$ and $\mu_{1} \in \left[\mu_{2}-\frac{1}{4}, \mu_{2}+\frac{1}{4}\right]$.

\item $\left[\mu_{2}-\frac{1}{4}, \frac{3}{4}\right]$ if $\mu_{2} \in \left[\frac{3}{4},1\right]$ and $\mu_{1} \in \left[\mu_{2}-\frac{1}{4}, \frac{3}{4}\right]$.

\item $\left[\frac{1}{4}, \mu_{2}+\frac{1}{4}\right] \cup \left\{\mu_{2}-\frac{1}{4}\right\}$ if $\mu_{2} \in \left[\frac{1}{4},\frac{1}{2}\right]$ and $\mu_{1} \in \left[\mu_{2}-\frac{1}{4}, \mu_{2}+\frac{1}{4}\right]$.

\item $\left[\frac{1}{4}, \mu_{2}+\frac{1}{4}\right]$ if $\mu_{2} \in \left[0, \frac{1}{4}\right]$ and $\mu_{1} \in \left[\frac{1}{4}, \mu_{2}+\frac{1}{4}\right]$.

\end{enumerate}
\end{lemma}

Note that we have not yet determined which sender should be visited first. Our first step is to show that if $\mu_{1}$ and $\mu_{2}$ satisfy one of the conditions for Lemma \ref{diffmean1} when $\mu_{2}$ is visited second, then they satisfy one of the conditions for Lemma \ref{diffmean1} when $\mu_{2}$ is visited first. Formally,

\begin{lemma}\label{exchangeability}
One of the four parametric restrictions in Lemma \ref{diffmean1} holds if and only if one of the four parametric restrictions in Lemma \ref{diffmean1} holds in which $\mu_{1}$ and $\mu_{2}$ are replaced with each other.
\end{lemma}

Assuming that one of the parametric conditions hold, the second step is to show that the receiver's expected payoff under any optimal protocol in which we assume that sender $1$ is visited first is the same as her expected payoff under any optimal protocol in which we assume sender $2$ is visited first. Hence, it does not matter which sender she visits first, and so she can break ties in that manner in any way that she chooses. 

This step requires just a couple sentences to prove: in each of the four cases described in Lemma \ref{diffmean1}, there is a stage $1$ optimal distribution in which the receiver learns nothing at the first sender. Her expected payoff under the optimal search protocol is thus
\[\mu_{1}^{2} + \mu_{2}^{2} + \frac{\mu_{1} + \mu_{2}}{2} - 2\mu_{1}\mu_{2} + \frac{1}{16}\]
which is invariant to an exchange of $\mu_{1}$ and $\mu_{2}$. Finally, we arrive at the heterogeneous means analog to Proposition \ref{main}:

\begin{proposition}\label{maindiff}
There is a full information equilibrium if $|\mu_{2} - \mu_{1}| \leq \frac{1}{4}$ and
\begin{enumerate}
    \item $\mu_{1}, \mu_{2} \in \left[\frac{1}{4}, \frac{3}{4}\right]$; or
    \item $\mu_{i} \leq \frac{3}{4} \leq \mu_{j}$ for $i, j \in \left\{1, 2\right\}$ and $i \neq j$; or
    \item $\mu_{i} \leq \frac{1}{4} \leq \mu_{j}$ for $i, j \in \left\{1, 2\right\}$ and $i \neq j$.
\end{enumerate}
\end{proposition}

\subsection{Sender Experiment-Dependent Cost Function}

Here, we modify the model by allowing the receiver's information processing cost at a sender to itself depend on the distribution chosen by the sender. That is, we amend the attention cost so that it is now given by (at sender $1$)

\begin{equation}\label{costfunctionex} C(q_{1}, p_{1})= k(p_{1})\int_{[0,1]}(x-\mu)^2 dq_{1}(x),\end{equation}
where $p_{1}$ is the choice of distribution of posterior beliefs by sender $1$ and $q_{1}$ is the garbling of that distribution chosen by the receiver. We assume that $k$ is weakly decreasing in the Blackwell order: the more informative the sender's experiment, the less costly a given information structure is for the receiver. This corresponds to the following intuition: the less informative the seller is, the costlier it is for the buyer to maintain a particular information structure, since she is forced to pay closer attention.\footnote{We thank the associate editor for suggesting such a cost of information.} 

We define $k_{F} \coloneqq k(p_{1}^{B})$ to be the minimum cost parameter, where $p_{1}^{B}$ is the Bernoulli distribution begotten by full information provision by the sender. Naturally, we stipulate that $k_{F}$ is non-negative.

With this modified cost function, our main result continues to hold. Namely, provided the attention cost is sufficiently high, there is an equilibrium in which both senders offer full information:

\begin{proposition}[Full information equilibrium]\label{mainw}
For all $k_{F} > \frac{1}{2}$, if $\mu \in \left[\frac{1}{4k_{F}},1-\frac{1}{4k_{F}}\right]$ then there \emph{is} an equilibrium in which both sellers offer full information.
\end{proposition}
\begin{proof}
On path, where each sender provides full information, the analysis is unchanged from earlier sections (with $k_{F}$ \textit{in lieu} of $k$), and the searcher's optimal protocol is unaltered. Moreover, should a sender deviate, then again the searcher can behave optimally by learning nothing at the deviating sender, eliminating the possibility for a sender to deviate profitably.
\end{proof}

\section{Conclusion}\label{conclusion}

We study a model of information disclosure by two senders who compete to persuade a receiver. The receiver, instead of passively accepting the experiment adopted by a sender, may choose to garble it before drawing a belief. The lower the informativeness of the chosen garbling, the lower her attention costs are.

We show how for a large class of parameters, it is an equilibrium for the senders to offer at least as much information to the receiver as she would choose for herself, if she could control information provision. In particular, full disclosure by both senders is an equilibrium. Moreover, there is no binary symmetric equilibrium (for any value of parameters) that does not give the receiver this first best outcome. We prove robustness to various modeling assumptions.

This is the result of an interesting trade-off that generalizes beyond the specifics of our model. Due to attention costs, the receiver never finds it worthwhile to learn either sender's type perfectly. That is, even with access to full information, she leaves some scope for further learning about each. Moreover, since her task is to choose between the senders, information on the quality of one sender partially substitutes for information on the quality of the other. For example, learning a lot about the quality of one drug on the market can be just as good (for the accuracy of a doctor's decision) as learning a bit about both alternatives.

Consequently, starting from a situation of full disclosure by both senders, if either sender deviates and restricts the receiver's learning, she has an opportunity to make up for it by using some of the `surplus' information--so far unused--about the \emph{other} sender. The deviating sender thus has limited ability--if any--to affect the overall quality of the receiver's information across the two alternatives.

This channel clearly breaks down in the absence of attention costs (so that the receiver always uses all available information), or if there is only one sender (so that there is no notion of substitutability). Our model identifies novel strategic incentives for greater information disclosure.

We motivated our study with the example of pharmaceutical companies strategically disclosing information to prescribing physicians. The assumption of high attention costs, as well as a low outside option for the receiver, are reasonable in this context. 

The model is well suited to study strategic disclosure in numerous other settings where information is `complex', e.g. the disclosure of features of retirement savings plans to consumers, or the informational content of political campaigns.

\bibliographystyle{econ}

\bibliography{ref}

\appendix
\section{Proofs}\label{appendixproofs}
Consider any $k>0$ and $\mu \in (0,1)$. Let each sender offer support $\left\{l,h\right\}$, with $l \in [0,\mu)$ and $h \in (\mu,1]$.

We begin by proving a series of Lemmata.

\begin{lemma}\label{stage2k}
Suppose that $R$'s stage 1 draw is $x \in [l,h]$ and she visits the sender at stage 2. $R$'s stage 2 optimal garbling is either degenerate or binary, and its support is as follows.
\begin{enumerate}

\item If $k>\frac{1}{2(h-l)}$ and $\mu \leq \min\left\{h-\frac{1}{2k},l+\frac{1}{2k}\right\}$:

\[\begin{cases} \left\{\mu\right\} & if \ x \in [l,l+k(\mu-l)^2] \\ \left\{l,l+\sqrt{\frac{x-l}{k}}\right\} & if \ x \in (l+k(\mu-l)^2,l+\frac{1}{4k}) \\ \left\{x-\frac{1}{4k},x+\frac{1}{4k}\right\} & if \ x \in \left[l+\frac{1}{4k},\mu+\frac{1}{4k}\right) \\ \left\{\mu\right\} & if \ x \in \left[\mu+\frac{1}{4k},h\right]
\end{cases}\]

\item If $k>\frac{1}{2(h-l)}$ and $\mu \geq \max\left\{h-\frac{1}{2k},l+\frac{1}{2k}\right\}$:

\[left[\begin{cases} \left\{\mu\right\} &if \  x \in \left[l,\mu-\frac{1}{4k}\right]  \\ \left\{x-\frac{1}{4k},x+\frac{1}{4k}\right\} & if \   x \in \left(\mu-\frac{1}{4k},h-\frac{1}{4k}\right]\\  \left\{h-\sqrt{\frac{h-x}{k}},h\right\} & if \  x \in \left(h-\frac{1}{4k},h-k(h-\mu)^2\right)\\ \left\{\mu\right\} & if  x \in [h-k(h-\mu)^2,h] \end{cases}\]

\item If $l+\frac{1}{2k} \leq \mu \leq h-\frac{1}{2k}$:

\[\begin{cases} \left\{x-\frac{1}{4k},x+\frac{1}{4k}\right\} & if \ x \in \left(\mu-\frac{1}{4k},\mu+\frac{1}{4k}\right) \\ \left\{\mu\right\} & if \ x \in \left[l,\mu-\frac{1}{4k}\right] \cup \left[\mu+\frac{1}{4k},h\right]
\end{cases}\]

 \item If $k>\frac{1}{2(h-l)}$ and $h-\frac{1}{2k}<\mu<l+\frac{1}{2k}$:

\[\begin{cases} \left\{\mu\right\} & if \ x \in [l,l+k(\mu-l)^2]  \\ \left\{l,l+\sqrt{\frac{x-l}{k}}\right\} & if \ x \in (l+k(\mu-l)^2,l+\frac{1}{4k}) \\ \left\{x-\frac{1}{4k},x+\frac{1}{4k}\right\} & if \ x \in \left[l+\frac{1}{4k},h-\frac{1}{4k}\right] \\ \left\{h-\sqrt{\frac{h-x}{k}},h\right\} & if \ x \in (h-\frac{1}{4k},h-k(\mu-h)^2) \\ \left\{\mu\right\} & if \ x \in [h-k(\mu-h)^2,h]
\end{cases}\]

\item If $k \leq \frac{1}{2(h-l)}$:
\[\begin{cases} \left\{\mu\right\} &  if \ x \in [l, l+k(\mu-l)^2] \\ \left\{l,l+\sqrt{\frac{x-l}{k}}\right\} & if \ x \in (l+k(\mu-l)^2,l+k(h-l)^2)\\ \left\{l,h\right\} & if \ x \in [l+k(h-l)^2,h-k(h-l)^2] \\ \left\{h-\sqrt{\frac{h-x}{k}},h\right\} & if \ x \in (h-k(h-l)^2,h-k(\mu-h)^2) \\ \left\{\mu\right\} & if \  x \in [h-k(\mu-h)^2  h]  \end{cases}\]

\end{enumerate}
\end{lemma}

\begin{proof}
$R$'s stage 2 payoffs for a stage 2 belief $y$ are given by
\[U_2(y;x)=\max\left\{x,y\right\}-k(x-\mu)^2.\]
This is piecewise concave. We first obtain the concavification of $U_2(y;x)$ over $[l,h]$ and then use it to find the optimal garbling.

The concavification of $U_2(y;x)$ is obtained by joining two points $y_1,y_2$ (in a straight line) with $l \leq y_1<x<y_2 \leq h$. By the definition of concavification of a function, we must have\footnote{The best way to see this is to assume it is does not hold and see that the definition of concavification is violated.} \begin{equation} \label{equation} U'_2(y_1;x) \leq \frac{U_2(y_2;x)-U_2(y_1;x)}{y_2-y_1} \leq U'_2(y_2;x),\end{equation} with the first inequality holding with equality if $y_1>l$ and the second one holding with equality if $y_2<h$.

The solution to Inequation \ref{equation} with both equalities is \[y_1=x-\frac{1}{4k}, \ y_2=x+\frac{1}{4k}.\]

If $l+\frac{1}{4k} < x < h-\frac{1}{4k}$, the concavification is given by $y_1=x-\frac{1}{4k}, \ y_2=x+\frac{1}{4k}$.\\

If $x \leq \min\left\{l+\frac{1}{4k},h-\frac{1}{4k}\right\}$, the lower bound $l$ binds and the concavification has $y_1=l$. $y_2=l+\sqrt{\frac{x-l}{k}}$ is obtained from the second equality in Inequation \ref{equation}. \\

If $x \geq \max\left\{h-\frac{1}{4k},l+\frac{1}{4k}\right\}$, the upper bound $h$ binds and the concavification has $y_2=h$. $y_1=h-\sqrt{\frac{h-x}{k}}$ is obtained from the first equality in Inequation \ref{equation}.\\

If $h-\frac{1}{4k}<x<l+\frac{1}{4k}$, the concavification is:
\begin{enumerate}
    \item $y_1=l$, $y_2=l+\sqrt{\frac{x-l}{k}}$ if $l+\sqrt{\frac{x-l}{k}} \leq h$.
    \item $y_2=h$, $y_1=h-\sqrt{\frac{h-x}{k}}$ if $h-\sqrt{\frac{h-x}{k}} \geq l$.
    \item $y_1=l$, $y_2=h$ otherwise.
\end{enumerate}

Having obtained the concavification for any $x$, the optimal stage 2 garbling has support $\left\{y_1,y_2\right\}$ if $\mu \in (y_1,y_2)$, and support $\left\{\mu\right\}$ otherwise. Straightforward algebra then gives us the stated result.
\end{proof}

\begin{lemma}\label{stage1k}
\begin{enumerate}
    \item  If $k>\frac{1}{2(h-l)}$ and $\mu \leq \min\left\{h-\frac{1}{2k},l+\frac{1}{2k}\right\}$,  $R$'s optimal stage 1 garbling is
    \begin{enumerate}
        \item Any Bayes plausible distribution with support drawn from the set $\left\{\mu-\frac{1}{4k}\right\} \cup \left[l+\frac{1}{4k},\mu+\frac{1}{4k}\right]$ if $\mu \geq l+\frac{1}{4k}$.
        \item The distribution with support $\left\{l,y_1(\mu)\right\}$ with $y_1(\mu) \in \left(\mu,l+\frac{1}{4k}\right)$ if $\mu<l+\frac{1}{4k}$.
    \end{enumerate}

\item If $k>\frac{1}{2(h-l)}$ and $\mu \geq \max\left\{h-\frac{1}{2k},l+\frac{1}{2k}\right\}$, $R$'s optimal stage 1 garbling is:
\begin{enumerate}
\item  Any Bayes plausible distribution with support drawn from the set $\left[\mu-\frac{1}{4k},h-\frac{1}{4k}\right] \cup \left\{\mu+\frac{1}{4k}\right\}$ if $\mu \leq h-\frac{1}{4k}$.
\item The distribution with support $\left\{y_2(\mu),h\right\}$ with $y_2(\mu) \in \left(h-\frac{1}{4k},\mu\right)$ if $h-\frac{1}{4k}<\mu$.
 \end{enumerate}
 \item  If $l+\frac{1}{2k} \leq \mu \leq h-\frac{1}{2k}$, $R$'s optimal stage 1 garbling is any Bayes plausible distribution with support on  $\left[\mu-\frac{1}{4k},\mu+\frac{1}{4k}\right]$.
 \item If $k>\frac{1}{2(h-l)}$ and $h-\frac{1}{2k}<\mu<l+\frac{1}{2k}$, $R$'s stage 1 optimal garbling is:
 \begin{enumerate}
     \item Any Bayes plausible distribution with support drawn from $\left\{\mu-\frac{1}{4k}\right\} \cup \left[l+\frac{1}{4k},h-\frac{1}{4k}\right]\cup \left\{\mu+\frac{1}{4k}\right\}$ if $l+\frac{1}{4k} \leq \mu \leq h-\frac{1}{4k}$.
     \item The distribution with support $\left\{l,y_1(\mu)\right\}$ with $y_1(\mu) \in \left(\mu,l+\frac{1}{4k}\right)$ if $\mu<l+\frac{1}{4k}$.
     \item The distribution with support $\left\{y_2(\mu),h\right\}$ with $y_2(\mu) \in \left(h-\frac{1}{4k},\mu\right)$ if $h-\frac{1}{4k}<\mu$.
 \end{enumerate}
 
 \item If $k \leq \frac{1}{2(h-l)}$, then
 \begin{enumerate}
     \item If $\mu \leq \frac{l+h}{2}$, $R$'s optimal stage 1 garbling is $\left\{l,y_1(\mu)\right\}$, where $y_1(\mu)>\mu$ is either on $\left(l+k(\mu-l)^2,l+k(h-l)^2\right)$ or on $\left[l+k(h-l)^2,h-k(h-l)^2\right]$.
 \item If $\mu>\frac{l+h}{2}$, $R$'s optimal stage 1 garbling is $\left\{y_2(\mu),h\right\}$, where $y_2(\mu)<\mu$ is either on $\left[l+k(h-l)^2,h-k(h-l)^2\right]$ or on $\left(h-k(h-l)^2,h-k(\mu-h)^2\right)$. 
 
 \end{enumerate}

 \end{enumerate}
\end{lemma}
\medskip

\begin{proof}
Let $U_1(x)$ be $R$'s first stage continuation payoffs for a first stage belief $x$. Say the stage 2 distribution following $x$ has support $\left\{y_1,y_2\right\}$, with $y_1 \leq y_2$ and $\nu y_1+(1-\nu)y_2=\mu$. Then $U_1(x)=\nu U_2(y_1;x)+ (1-\nu) U_2(y_2;x)-k(x-\mu)^2$. The concavification of $U_1$ over $[l,h]$ is used to obtain the stage 1 optimal distribution.

For any $\mu$, $U_1$ is continuous. Note that $U_1$ is affine over any interval of $x$ for which the stage 2 optimal garbling is $\left\{x-\frac{1}{4k},x+\frac{1}{4k}\right\}$.\\

\begin{remark}
If the stage 1 optimal garbling is unique, then it cannot have support $\left\{\mu\right\}$.
\end{remark}

The reason for this is the following. If the stage 1 unique optimal garbling is degenerate, then it is verified from Lemma \ref{stage2k} that the stage 2 optimal garbling has binary support, say $\left\{y_1,y_2\right\}$. But then, choosing the garbling $\left\{y_1,y_2\right\}$ at stage 1 and $\left\{\mu\right\}$ at stage 2 must give the same expected payoff, and hence must be optimal. This is a contradiction.\\

Now, first let $k>\frac{1}{2(h-l)}$ and $\mu \leq \min\left\{h-\frac{1}{2k},l+\frac{1}{2k}\right\}$.

Then $U_1$ is strictly convex in a right neighborhood of $k(\mu-l)^2$ and concave everywhere else (weakly on $\left(l+\frac{1}{4k},\mu+\frac{1}{4k}\right)$). Then, the concavification must join points $z_1 \leq k(\mu-l)^2$ and $z_2 > k(\mu-l)^2$ (in a straight line), with $z_1,z_2$ determined by a condition analogous to Inequation \ref{equation}.

Say $\mu \geq l+\frac{1}{4k}$. Then it is verified that $z_1=\mu-\frac{1}{4k}$ and $z_2=l+\frac{1}{4k}$. Since $\mu \in \left[l+\frac{1}{4k},\mu+\frac{1}{4k}\right]$ and $U_1$ is affine over this interval, a distribution with support on $\left\{\mu-\frac{1}{4k}\right\} \cup \left[l+\frac{1}{4k},\mu+\frac{1}{4k}\right]$ would be optimal.

Now say $\mu<l+\frac{1}{4k}$. Clearly the lower bound $l$ would bind and $z_1=l$ must hold. $z_2$ is obtained from the second equality in Inequation \ref{equation}, and it must be higher than $\mu$, since otherwise the optimal garbling would uniquely be degenerate, and we ruled that out above. $z_2$ is denoted by $y_1(\mu)$ in the statement of the Lemma. \\

Now let $k>\frac{1}{2(h-l)}$ and $\mu \geq \max\left\{h-\frac{1}{2k},l+\frac{1}{2k}\right\}$. The argument is symmetric to the preceding one.

In this case $U_1$ is strictly convex in a left neighborhood of  $h-k(h-\mu)^2$ and concave everywhere else (weakly on $(\mu-\frac{1}{4k},h-\frac{3}{4k}))$. The concavification is obtained by joining points $z_1$ and $z_2$ as before.

It is verified that for $\mu \leq h-\frac{1}{4k}$, $z_1=h-\frac{1}{4k}$ and $z_2=\mu+\frac{1}{4k}$. This tells us that a distribution with support on $\left[\mu-\frac{1}{4k},h-\frac{1}{4k}\right] \cup \left\{\mu+\frac{1}{4k}\right\}$ would be optimal.

For $\mu>h-\frac{1}{4k}$, $z_2=h$ must hold. Now $z_1$ is found from the first equality in Inequation \ref{equation}, and it must be lower than $\mu$, since otherwise the stage 1 optimal garbling would uniquely be degenerate. $z_1$ is denoted by $y_2(\mu)$ in the statement of the Lemma.\\

Cases \emph{3} and \emph{4} are dealt with completely analogously.\\

Finally, let $k \leq \frac{1}{2(h-l)}$.

Then $U_1$ is strictly convex in a right neighborhood of $l+k(\mu-l)^2$, and in a left neighborhood of $h-k(h-\mu)^2$, and \emph{strictly} concave everywhere else.

Clearly, the concavification must:
\begin{enumerate}
\item join points $z_1 \in [l,l+k(\mu-l)^2)$ and $z_2>l+k(\mu-l)^2$ in a straight line, and
\item join points $z_3<h-k(h-\mu)^2$ and $z_4 \in (h-k(h-\mu)^2,h]$ in a straight line. 
\end{enumerate}

As usual, these points are determined by a condition analogous to Inequation \ref{equation}. It turns out that $z_1=l$ and $z_4=h$, while the positions of $z_2$ and $z_3$ depend on parameters. The optimal garbling is either $\left\{l,z_2\right\}$ or $\left\{z_3,h\right\}$, depending on where $\mu$ lies.

\end{proof}
\medskip

The previous result immediately gives us the following useful corollary.

\begin{corollary}\label{cor} 
The following two statements are equivalent:
\begin{enumerate}
    \item $\mu \in [l+\frac{1}{4k},h-\frac{1}{4k}]$ and $k>\frac{1}{2(h-l)}$.
    \item There are multiple stage 1 optimal garblings for $R$, including support $\left\{\mu\right\}$ and support $\left\{\mu-\frac{1}{4k},\mu+\frac{1}{4k}\right\}$.
\end{enumerate}
\end{corollary}
\medskip

\begin{lemma}\label{nodeviation}
Suppose $k>\frac{1}{2(h-l)}$ and $\mu \in [l+\frac{1}{4k},h-\frac{1}{4k}]$. If $R$'s behavior is as specified in Lemmata \ref{stage2k} and \ref{stage1k}, then conditional on being the first sender to be visited, the probability of being selected is the same regardless of which stage 1 optimal garbling is chosen by $R$.
\end{lemma}
\medskip

\begin{proof}
We show the proof for $\mu \leq \min\left\{h-\frac{1}{2k},l+\frac{1}{2k}\right\}$. It is entirely analogous for the other cases from Lemma \ref{stage2k}.

Suppose $l+\frac{1}{4k} \leq \mu \leq \min\left\{h-\frac{1}{2k},l+\frac{1}{2k}\right\}$ and $R$'s first stage response is a distribution $F$ on $\left\{\mu-\frac{1}{4k}\right\} \cup [l+\frac{1}{4k},\mu+\frac{1}{4k}]$.

Using Lemma \ref{stage2k} it is easy to see that the probability of the first sender being selected conditional on a first stage belief $x$ is given by

\[P(x)= \begin{cases} 0 & if \ x=\mu-\frac{1}{4k} \\
2kx-2k\mu+\frac{1}{2} & if \ x \in [l+\frac{1}{4k},\mu+\frac{1}{4k}]
\end{cases}\]

Suppose that $F$ places a mass $p \geq 0$ on $\mu-\frac{1}{4k}$. Then conditional on being visited first, a sender's expected probability of being selected is given by 

\begin{equation} V_1=p*0+\int_{l+\frac{1}{4k}}^{\mu+\frac{1}{4k}} P(x) \ dF(x)\end{equation}

 Next note that \begin{equation} p(\mu-\frac{1}{4k})+\int_{l+\frac{1}{4k}}^{\mu+\frac{1}{4k}} x dF(x)=\mu \end{equation}

and 

\begin{equation}\int_{l+\frac{1}{4k}}^{\mu+\frac{1}{4k}} dF(x)=1-p\end{equation}
\medskip

Inserting Equations 3 and 4 into Equation 2, we get that $V_1=\frac{1}{2}$, which is independent of $F$.
\end{proof}

\subsection{Proof of Proposition \ref{binarygeneral}}\label{prooflong}

Suppose each sender offers support $\left\{l,h\right\}$, with $l \in [0,\mu)$ and $h \in (\mu,1]$.\\

First let $k>\frac{1}{2(h-l)}$ and $\mu \in [l+\frac{1}{4k},h-\frac{1}{4k}]$. 

Given a stage 1 draw $x$, $R$'s optimal stage 2 garbling is specified in Lemma \ref{stage2k}. If this garbling does not have support $\left\{\mu\right\}$, $R$ necessarily visits the second sender. If it is $\left\{\mu\right\}$, she is indifferent between visiting him and not, and may choose either way.

At stage 1, she has multiple best responses. The most informative one among them has support $\left\{\mu-\frac{1}{4k},\mu+\frac{1}{4k}\right\}$, and from Lemma \ref{stage1k} it is the only one that is necessarily followed by no learning at stage 2. We assume that she breaks her indifference in favor of this distribution.

At belief $\mu-\frac{1}{4k}$ she accepts the first sender with certainty, and at belief $\mu+\frac{1}{4k}$  accepts the other one with certainty.

Then if a sender deviates to a different distribution, his payoffs may be affected only if he is visited first and the distribution he deviates to is such that $\left\{\mu-\frac{1}{4k},\mu+\frac{1}{4k}\right\}$ is not a garbling of it.

In this case, regardless of the deviation, $R$ can secure a payoff equal to what she gets in the absence of the deviation, by picking $\left\{\mu\right\}$ at stage 1, followed by visiting the other sender and choosing $\left\{\mu-\frac{1}{4k},\mu+\frac{1}{4k}\right\}$. Thus the deviation cannot force $R$ to choose from outside the set of optimal garblings from Lemma \ref{stage1k}.

But then due to Lemma \ref{nodeviation}, the deviating sender's payoffs are unaffected. Thus, there does not exist a profitable deviation and we have an equilibrium.\\

Next say that either $k>\frac{1}{2(h-l)}$ and $\mu \not \in [l+\frac{1}{4k},h-\frac{1}{4k}]$, or $k \leq \frac{1}{2(h-l)}$.

Then from Lemmata \ref{stage1k} and \ref{stage2k}, $R$ chooses a unique binary garbling at stage 1, and exactly one belief in the support is followed by a visit to the second sender.

Denote the stage 1 belief following which $R$ does learn at stage 2 by $w$. Under each possibility we show that there is a profitable deviation for a sender.

\emph{Possibility 1}: If $w<\mu$, then $\exists \ l' \in [\max\left\{0,\mu-\frac{1}{2k}\right\},\mu)$ s.t $w=l'+k(\mu-l')^2$. Suppose a sender deviates to support $\left\{l',h\right\}$. Then from Lemma \ref{stage2k}, if the deviating sender is visited second, $R$ chooses support $\left\{\mu\right\}$ and selects the deviating sender with certainty. This does not affect $R$'s behavior if the deviating sender is visited first, since $l'<w$. Thus the sender profits from this deviation.

\emph{Possibility 2}: If $w>\mu$ and is followed by a stage 2 best response $\left\{l,l+\sqrt{\frac{w-l}{k}}\right\}$, then it must be true that $w \in (l+k(\mu-l)^2,l+\frac{1}{4k})$. $\exists \ h'<l+\sqrt{\frac{w-l}{k}}$ s.t $w<h'-k(h'-l)^2$. Suppose a sender deviates to $\left\{l,h'\right\}$. Then $k \leq \frac{1}{2(h'-l)}$, and Lemma \ref{stage2k} tells us that if the deviating sender is visited at stage 2, $R$'s response changes to $\left\{l,h'\right\}$. $w<h'<l+\sqrt{\frac{w-l}{k}}$ implies that this is profitable if visited at stage 2, without affecting what happens if visited at stage 1. Thus the deviation is profitable.

\emph{Possibility 3}: If $w>\mu$ and is followed by a stage 2 best response of $\left\{l,h\right\}$. Then it is seen that $w<h$, so that $h$ does not bind at stage 1. This implies that a sender can increase or decrease $h$ slightly to $h'$ (so that $\left\{l,h'\right\}$ is instead chosen), without affecting what happens if he is visited first. This is clearly profitable. \hfill $\qed$

\subsection{Proof of Claim \ref{firstbestclaim}}
\emph{`Only if'}:  Suppose that there is no  equilibrium in which both senders offer full info. Then, Proposition \ref{main} tells us that either $k \leq \frac{1}{2}$, or $k>\frac{1}{2}$ and $\mu \not \in [\frac{1}{4k},1-\frac{1}{4k}]$. Lemma \ref{stage2k} and Lemma \ref{stage1k} tell us $R$'s unique best response (on path) to full info from both senders. Now we need to show that there is no equilibrium where she gets her first best payoff. For the sake of contradiction, suppose that there is such an equilibrium--and where sender $i$ offers some $p_i$. From the discussion in the main text, this just means that $R$'s best response on path to $(p_1,p_2)$, is the same as the best response to full information. We argue, however, that the same deviations that we identified for full info, also work for this supposed equilibrium. Recall the nature of those deviations from \ref{prooflong}: they do not make a difference if the deviating sender is visited first, and restrict learning if visited second. Now if $p_1,p_2$ is the equilibrium under consideration and the same deviation occurs, $R$'s response to this deviation would be as under full info: if she visits the deviating sender first, she would realize she can continue to choose as on path; if she visits him second, she would make the same adjustment as under the full info scenario.

Thus, since the deviation was profitable under full info, it must be profitable here, and $p_1,p_2$ cannot be an equilibrium. \hfill \qed

\subsection{Proof of Proposition \ref{main}}
See the proof of Proposition \ref{binarygeneral}, setting $l=0,h =1$.

\subsection{Proof of Proposition \ref{kzero}}
Existence of the uninformative equilibrium is proven in the text. Here we show non-existence of a full information equilibrium.

Suppose that each sender chooses a fully informative distribution. Because each sender has chosen the same distribution (on path), $R$ is indifferent as to whom she visits first. Hence, suppose that she visits sender $1$ first with probability $\lambda \in [0,1]$ and sender $2$ with its complement.

If sender $1$ is visited first, then upon $R$'s visit, $1$ is realized with probability $\mu$. At this point, she will stop and select sender $1$. On the other hand, if $0$ is realized then she will select sender $2$ without visiting. The symmetric statements hold for sender $2$ and her payoff is
\[u_{2} = \lambda (1-\mu) + (1-\lambda)\mu\]

Now suppose that sender $2$ deviates and chooses a distribution that consists of 1 with probability $\eta \coloneqq \mu - \frac{1}{n}$, $n \in \mathbb{N}$, $n > \frac{1}{\mu}$, and $\epsilon$ with probability $1 + \frac{1}{n} - \mu$, where $\epsilon \coloneqq \frac{1}{n+1 - \mu n}$. If sender $1$ is visited first then again sender $2$ obtains an expected payoff of $(1-\mu)$. If sender $2$ is visited first, with probability $\eta$, $1$ is realized and sender $2$ is selected and with probability $(1-\eta)$, $\epsilon$ is realized. At this point $R$ visits sender $1$ and obtains a realization of $0$ with probability $1-\mu$, at which point she selects sender $2$. Accordingly,
\[u_{2} = \lambda (1-\mu) + (1-\lambda)\left(\eta + (1-\eta)(1-\mu)\right)\]
and so sender $2$ has a profitable deviation if and only if
\[\lambda (1-\mu) + (1-\lambda)\left(\eta + (1-\eta)(1-\mu)\right) > \lambda (1-\mu) + (1-\lambda)\mu\]
which reduces to
\[\frac{1 + 2n - \sqrt{1 + 4n}}{2n} > \mu\]
provided $\lambda < 1$. Without loss of generality we may assume this, since otherwise the same argument would suffice for a deviation by sender $1$.

The limit of the left hand side goes to $1$ as $n$ goes to $\infty$; hence for any $\mu < 1$ there exists a $\hat{n}$ such that the left hand side is strictly greater than $\mu$ for all $n > \hat{n}$. We conclude that for any $\mu < 1$ there exists a profitable deviation, negating the possibility that full information is an equilibrium. \hfill $\qed$

\subsection{Proof of Claim \ref{kzeroother}}
\textbf{For \boldmath $\mu \leq \frac{1}{2}$ \unboldmath}:

Let each sender choose the uniform distribution on $[0,2\mu]$, and suppose that $R$ visits sender $1$ first with probability $\lambda \in [0,1]$ and sender $2$ with its complement.

No matter the realization at stage 1, $R$ will proceed and visit the other sender as well before selecting one of them. Hence, $u_{1} = u_{2} = \frac{1}{2}$. Next, we check for a profitable deviation. Suppose sender $1$ deviates to a distribution that contains a probability measure of size $a$ on $[2\mu, 1]$ and some portion $F$ on $[0,2\mu)$. It is clear that it is without loss of generality to set $a$ to be a point mass on $2\mu$.

If sender $1$ is visited first then with probability $a$, he is selected and sender $2$ is never visited; and otherwise, sender $2$ is visited after which $R$ selects the sender with the highest realization. If sender $2$ is visited first, then no matter what, sender $1$ is also visited, after which the comparison ensues.

Sender $1$'s payoff is
\[u_{1} = \lambda\left(a + \int_{0}^{2\mu}\int_{0}^{x}dG(y)dF(x)\right) + \left(1-\lambda\right)\left(a + \int_{0}^{2\mu}\int_{0}^{x}dGdF\right) = a + \int_{0}^{2\mu}\int_{0}^{x}dG(y)dF(x)\]
where $G(y) = \frac{y}{2\mu}$ is the (on-path) distribution chosen by sender $2$ and where $\int_{0}^{2\mu}dF = 1-a$ and $\int_{0}^{2\mu}xdF = 2-2\mu a$.

Next, we use the result in \cite{whitmeyer2019mixtures} who establish that it suffices to show that $1$ has no profitable deviation to any binary distribution. Let $F$ be described by $\alpha$ with probability $p$ and $\beta$ with probability $1-p$; where $0 \leq \alpha \leq \mu$, $\mu \leq \beta \leq 2\mu$, and $\alpha p + \beta (1-p) = \mu$. Consequently, we rewrite $u_{1}$, which becomes
\[\begin{split}
    u_{1} &= (1-p)F(\beta) + p F(\alpha)\\
    &= (1-p)\frac{\beta}{2\mu} + p\frac{\alpha}{2\mu} = \frac{1}{2}
\end{split}\]
Hence, there is no profitable deviation.\hfill $\qed$\\

\begin{flushleft}
\textbf{For \boldmath $\mu>\frac{1}{2}$ \unboldmath}:
\end{flushleft}

On path, sender $1$'s payoff is
\[\begin{split}
    u_{1} &= \lambda\left(2 - \frac{1}{\mu} + \int_{0}^{2(1-\mu)}\int_{0}^{x}\frac{1}{2\mu}\frac{1}{2\mu}dydx\right) + \left(1-\lambda\right)\left(\left(\frac{1}{\mu}-1\right)\left(2 - \frac{1}{\mu}\right) + \int_{0}^{2(1-\mu)}\int_{0}^{x}\frac{1}{2\mu}\frac{1}{2\mu}dydx\right)\\
    &= \dfrac{2\left(2{\mu}-1\right)^2\lambda+\left(1-\mu\right)\left(3{\mu}-1\right)}{2{\mu}^2}
\end{split}\]
If sender $1$ deviates to $1$ with probability $\mu$ and $0$ with probability $1-\mu$, his payoff from deviating is
\[u_{1}^{D} = \lambda \mu + (1-\lambda)\left(\frac{1}{\mu}-1\right)\mu = 1 + 2\lambda \mu - \lambda - \mu\]
The difference, $u_{1}^{D} - u_{1}$ is 
\[\dfrac{\left(1-\mu\right)^2\left(2\mu-1\right)\left(2\lambda-1\right)}{2\mu^2}\]
Since $\mu > \frac{1}{2}$, this is positive provided $\lambda > \frac{1}{2}$ and negative provided $\lambda < \frac{1}{2}$. Thus, if $\lambda \neq \frac{1}{2}$ there exists a profitable deviation (if $\lambda < \frac{1}{2}$, sender $2$ can deviate profitably in the analogous fashion).

It remains to show that this vector of distributions is an equilibrium for $\lambda = \frac{1}{2}$. Substituting $\lambda = \frac{1}{2}$ into $u_{1}$, we see that $u_{1} = \frac{1}{2}$ on path. Just as for $\mu \leq \frac{1}{2}$, from \cite{whitmeyer2019mixtures} we need check only deviations to binary distributions. Let $F$ be described by $\alpha$ with probability $p$ and $\beta$ with probability $1-p$, where  $\alpha p + \beta (1-p) = \mu$ and $0 \leq \alpha \leq \mu$. There are two cases that we need to consider. 1. $\mu \leq \beta \leq 2(1-\mu)$; and 2. $\beta = 1$. In the first case,
\[\begin{split}
    u_{1} &= (1-p)F(\beta) + p F(\alpha) \\
    &= (1-p)\frac{\beta}{2\mu} + p\frac{\alpha}{2\mu} = \frac{1}{2}
\end{split}\]
and in the second case
\[\begin{split}
    u_{1} &= \frac{1}{2}\left(1-p + p F(\alpha)\right) + \frac{1}{2}\left(\left(\frac{1}{\mu}-1\right)(1-p) + p F(\alpha)\right) \\
    &= p\frac{\alpha}{2\mu} + \frac{1-p}{2\mu} = \frac{1}{2}
\end{split}\]
where we used the fact that $\beta = 1$ implies that $1-p = \mu - p\alpha$. Hence, there is no profitable deviation.\hfill $\qed$

\subsection{Proof of Lemma \ref{distristage2}}
See the proof of Lemma \ref{stage2k}, setting $l=0,h=1$ and $k=1$.

\subsection{Proof of Lemma \ref{distristage1}}
See the proof of Lemma \ref{stage1k}, setting $l=0,h=1$ and $k=1$.

\subsection{Proof of Lemma \ref{fullinfoeqm}}
See the proof of Lemma \ref{nodeviation}, setting $l=0,h=1$ and $k=1$.

\subsection{Proof of Claim \ref{obvious}}
Let $k>\frac{1}{2}$ and $\mu \in [\frac{1}{4k},1-\frac{1}{4k}]$.  As shown in Appendix \ref{prooflong}, one of $R$'s best responses to full information ($l=0, h=1$) from both senders is to choose the garbling $\left\{\mu-\frac{1}{4k},\mu+\frac{1}{4k}\right\}$ at stage 1 and to learn nothing at stage 2.

Suppose sender $i$ offers a distribution of which $\left\{\mu-\frac{1}{4k},\mu+\frac{1}{4k}\right\}$ is a garbling. Then, the aforementioned best response to full information is permissible, and thus continues to be a best response. Suppose $R$ chooses this response.

Then if a sender unilaterally deviates and is the one to be visited first, $R$ may respond by choosing $\left\{\mu\right\}$ and visiting the other sender, choosing $\left\{\mu-\frac{1}{4k},\mu+\frac{1}{4k}\right\}$ for him. Exactly as in the proof for existence of a full information equilibrium (Proposition \ref{binarygeneral} for $h=1,l=0$), Lemma \ref{nodeviation} can be used to argue that the deviation cannot be profitable. \hfill $\qed$

\subsection{Proof of Corollary \ref{uniformkpositive}}
\textbf{For \boldmath $\mu \leq \frac{1}{2}$}:

We show that $\left\{\mu-\frac{1}{4k},\mu+\frac{1}{4k}\right\}$ is a mean preserving contraction of the uniform distribution on $[0,2\mu]$ when $k \geq \frac{1}{2\mu}$.

Define $l(x)$ as
\[l(x) = \begin{cases}
0 & 0 \leq x<\mu - \frac{1}{4k}\\
\frac{1}{2}x-\frac{\mu}{2}+\frac{1}{8k} & \mu - \frac{1}{4k} \leq x < \mu + \frac{1}{4k}\\
x - \mu & \mu + \frac{1}{4k} \leq x \leq 1
\end{cases}\]
Define $j(x) \coloneqq \int_{0}^{x}G(t)dt$:

\[j(x) = \begin{cases}
\frac{x^2}{4\mu} & 0 \leq x< 2\mu\\
x - \mu & 2\mu \leq x \leq 1
\end{cases}\]
It suffices to show that $\mu > \frac{1}{4k}$, that $j(x) - l(x) = 0$ has at most one real root, and that $j\left(\mu + \frac{1}{4k}\right) > l\left(\mu + \frac{1}{4k}\right)$.

Set $j(x) = l(x)$, which holds if and only if
\[x = \frac{4\mu k \pm \sqrt{8k\mu\left(1-2k\mu\right)}}{4k}\]
This is imaginary if and only if
\[\begin{split}
    k &> \frac{1}{2\mu}
\end{split}\]
and has a unique root for $k = \frac{1}{2\mu}$ (at $\mu$). $\mu - \frac{1}{4k} \geq \frac{\mu}{2} > 0$ for $k \geq \frac{1}{2\mu}$. It remains to verify that $j\left(\mu + \frac{1}{4k}\right) > k\left(\mu + \frac{1}{4k}\right)$; but it is simple to verify that this must hold. Thus, if $k \geq \frac{1}{2\mu}$, we have the result. \\

\textbf{For \boldmath $\mu>\frac{1}{2}$ \unboldmath}:

The proof is analogous to the preceding one, with the exception that $k$ must be sufficiently large so that $\mu + \frac{1}{4k} \leq 1$. This holds if and only if $k \geq \frac{1}{4(1-\mu)}$. This constraint binds for $\mu \geq \frac{2}{3}$ and $k \geq \frac{1}{2\mu}$ binds for $\mu \leq \frac{2}{3}$. \hfill $\qed$

\subsection{Proof of Proposition \ref{publicexp}}

Suppose each sender offers support $\left\{l,h\right\}$, with $l \in [0,\mu)$ and $h \in (\mu,1]$.\\

Lemma \ref{stage2k} and Lemma \ref{stage1k} continue to describe on path behavior. Lemma \ref{nodeviation} still holds.\\

First let $k>\frac{1}{2(h-l)}$ and $\mu \in [l+\frac{1}{4k},h-\frac{1}{4k}]$. 

On path behavior is exactly as in the baseline model: visit any one sender, pick $\left\{\mu-\frac{1}{4k},\mu+\frac{1}{4k}\right\}$ and take a decision without learning from the other sender.

Then if a sender deviates to a different distribution, it would be observed. Then $R$ can simply respond by visiting the other, non-deviating sender, and picking $\left\{\mu-\frac{1}{4k},\mu+\frac{1}{4k}\right\}$ for him and taking a decision immediately. 

Due to Lemma \ref{nodeviation}, the deviating sender's payoffs are the same as on path. Thus, there does not exist a profitable deviation and we have an equilibrium.\\

Next say that either $k>\frac{1}{2(h-l)}$ and $\mu \not \in [l+\frac{1}{4k},h-\frac{1}{4k}]$, or $k \leq \frac{1}{2(h-l)}$.

Then from Lemmata \ref{stage1k} and \ref{stage2k}, on path $R$ chooses a unique binary garbling at stage 1, and exactly one belief in the support is followed by a visit to the second sender.

Denote the stage 1 belief following which $R$ does learn at stage 2 by $w$. Under each possibility we show that there is a profitable deviation for a sender.

\emph{Possibility 1}: Say $w<\mu$ and the stage 2 garbling is $\left\{l,h\right\}$. There must be a sender, say sender $i$, who is visited first with probability $<1$ on path. Suppose sender $i$ deviates to $\left\{l',h\right\}$, where $l<l'<w$. But on observing this deviation, $R$ would choose to visit sender $i$ first. By doing this she could get her first best. Thus, behavior is as on path, except that the order of visits is changed: sender $i$ is visited first with probability 1. It is easy to verify that the payoff from being visited first is $>\frac{1}{2}$ (i.e. higher than payoff from being visited second), which means that this increase in probability of being visited first is profitable.

\emph{Possibility 2}: Say $w<\mu$ and the stage 2 garbling is $\left\{h-\sqrt{\frac{h-k}{k}},h\right\}$. Everything is as in possibility 1, except that $l'$ is chosen such that $h-\sqrt{\frac{h-k}{k}}<l'<w$.

\emph{Possibility 3:} If $w>\mu$ and is followed by a stage 2 best response $\left\{l,l+\sqrt{\frac{w-l}{k}}\right\}$ or $\left\{l,h\right\}$. Then if a sender deviates to no information, clearly $R$ would just learn from the other sender with a threshold of acceptance $\mu$. It is verified that the deviating sender's payoffs then are higher than the payoffs on path, conditional on being visited first as well as conditional on being visited second.  \hfill $\qed$

\subsection{Proof of Lemma \ref{diffmean1}}

See the proof of Lemma \ref{stage1k}, setting $l=0,h=1$ and $k=1$ and using $\mu_{2}$ as the mean for the second sender and $\mu_{1}$ as the mean for the first sender.

\subsection{Proof of Lemma \ref{exchangeability}}

Let us begin by looking at the parametric conditions given in bullet points $1$ and $3$ of Lemma \ref{diffmean1}. By symmetry it suffices to assume that one of these two pairs of conditions holds for the scenario in which sender $2$ is visited second, and show that that implies that one of the four pairs of conditions for the scenario in which sender $2$ is visited first must hold. Observe that the conditions for bullet points $1$ and $3$ reduce to $|\mu_{1} - \mu_{2}| \leq \frac{1}{4}$ and $\mu_{2} \in \left[\frac{1}{4}, \frac{3}{4}\right]$. It is easy to see that if $\mu_{1} \in \left[\frac{1}{4}, \frac{3}{4}\right]$ then we are done. What if $\mu_{1} \notin \left[\frac{1}{4}, \frac{3}{4}\right]$? WLOG suppose that $\mu_{1} < \frac{1}{4}$. By assumption we must have $\mu_{2} - \frac{1}{4} \leq \mu_{1}$ and $\mu_{2} \geq \frac{1}{4}$. Hence, condition $4$ (with $\mu_{2}$ and $\mu_{1}$ transposed) must hold.

Next, we turn our attention to the conditions given in bullet points $2$ and $4$. WLOG it suffices to focus on the conditions in bullet point $2$. As we did in the previous paragraph, it suffices to assume that these conditions hold for the scenario in which sender $2$ is visited second, and show that that implies that one of the four pairs of conditions for the scenario in which sender $2$ is visited first must hold. By construction, $\mu_{2} \leq \mu_{1} + \frac{1}{4}$ and $\mu_{1} \in \left[\frac{1}{2}, \frac{3}{4}\right]$. Moreover, $\mu_{2} \geq \mu_{1} > \mu_{1} - \frac{1}{4}$, and so condition $1$ (with $\mu_{2}$ and $\mu_{1}$ transposed) must hold. \hfill $\qed$

\subsection{Proof of Proposition \ref{maindiff}}

It suffices to show that conditional on being the first sender to be visited, the probability of being selected is the same regardless of which stage $1$ optimal garbling is chosen by $R$. The remainder of the proof follows analogously to the proof of Lemma \ref{nodeviation}. Alternatively, observe that it follows from the fact that probability of the first sender being selected conditional on a first stage belief $x$ is either $0$, $1$, or a function that is affine in $x$. \hfill $\qed$

\end{document}